\pdfoutput=1
\documentclass[final,letterpaper,oneside,onecolumn,11pt]{IEEEtran}
\usepackage{latexsym,url}
\usepackage{amsmath,amssymb,graphicx}
\usepackage{citesort}
\usepackage{epsf}

\long\def\comment#1{}

\newfont{\bbb}{msbm10 scaled 700}

\newfont{\bb}{msbm10 scaled 1100}

\newcommand{\dv}{{\bf d}}
\newcommand{\ev}{{\bf e}}

\newcommand{\sv}{{\bf s}}
\newcommand{\tv}{{\bf t}}
\newcommand{\uv}{{\bf u}}

\newcommand{\vv}{{\bf v}}
\newcommand{\xv}{{\bf x}}
\newcommand{\yv}{{\bf y}}
\newcommand{\zv}{{\bf z}}

\renewcommand{\det}{{\hbox{det}}}

\usepackage{graphics} 
\usepackage{epsfig} 
\usepackage{amsmath} 
\usepackage{amssymb}  
\usepackage{dsfont}
\newtheorem{thm}{Theorem}
\newtheorem{cor}[thm]{Corollary}
\newtheorem{lem}[thm]{Lemma}

\title{ \LARGE \bf
Joint Physical Layer Coding and Network Coding for Bi-Directional Relaying
}

\author{Makesh Pravin Wilson, Krishna Narayanan, Henry Pfister and Alex Sprintson
\thanks{This work was supported by the National Science Foundation under grant CCR-0515296.  The authors are with the Department of Electrical and Computer Engineering, Texas A\&M University, College Station, TX 77843, USA {\tt\small \{makesh,krn,hpfister,spalex\}@ece.tamu.edu} }
}

\begin{document}

\maketitle

\begin{abstract}
We consider the problem of two transmitters wishing to exchange
information through a relay in the middle.  The channels between the
transmitters and the relay are assumed to be synchronized, average
power constrained additive white Gaussian noise channels with a real
input with signal-to-noise ratio (SNR) of {\sf snr}. An upper bound
on the capacity is $\frac12 \log (1+ {\sf snr})$ bits per
transmitter per use of the medium-access phase and broadcast phase
of the bi-directional relay channel.  We show that using lattice
codes and lattice decoding, we can obtain a rate of $\frac12 \log
(\frac12 + {\sf snr})$ bits per transmitter, which is essentially
optimal at high SNR. The main idea is to decode the sum of the
codewords modulo a lattice at the relay followed by a broadcast
phase which  performs Slepian-Wolf coding with structured codes. For
asymptotically low SNR, joint decoding of the two transmissions at
the relay (MAC channel) is shown to be optimal. We also show that if
the two transmitters use identical lattices with minimum angle
decoding, we can  achieve the same rate of $\frac12
\log(\frac12+{\sf snr})$. The proposed scheme can be thought of as a
joint physical layer, network layer code which outperforms other
recently proposed analog network coding schemes.
\end{abstract}

\maketitle
\section{Introduction, System Model and Problem statement}
\label{sec:introduction} We consider the bi-directional relaying
problem where two users try to exchange information with each other
through a relay in the middle. More specifically, we study a simple
3-node linear Gaussian network as shown in Fig.~\ref{fig:sysmodel}.
Nodes $A$ and $B$ wish to exchange information between each other
through the relay node $R$, however, nodes $A$ and $B$ cannot
communicate with each other directly. Let $\mathbf{u}_A \in
\{0,1\}^k$ and $\mathbf{u}_B \in \{0,1\}^k$, be the information
vectors at nodes $A$ and $B$ (vectors are denoted by bold face
letters such as $\uv$ throughout the paper). The information is
assumed to be encoded into vectors (codewords) $\mathbf{x}_1 \in
\mathbb{R}^n$ and $\mathbf{x}_2 \in \mathbb{R}^n$ at nodes $A$ and
$B$, respectively, and transmitted. We assume that communication
takes place in two phases - a multiple access (MAC) phase and a
broadcast phase, which are briefly described below.

\begin{figure}[h]\center
  \includegraphics[width=5.0in]{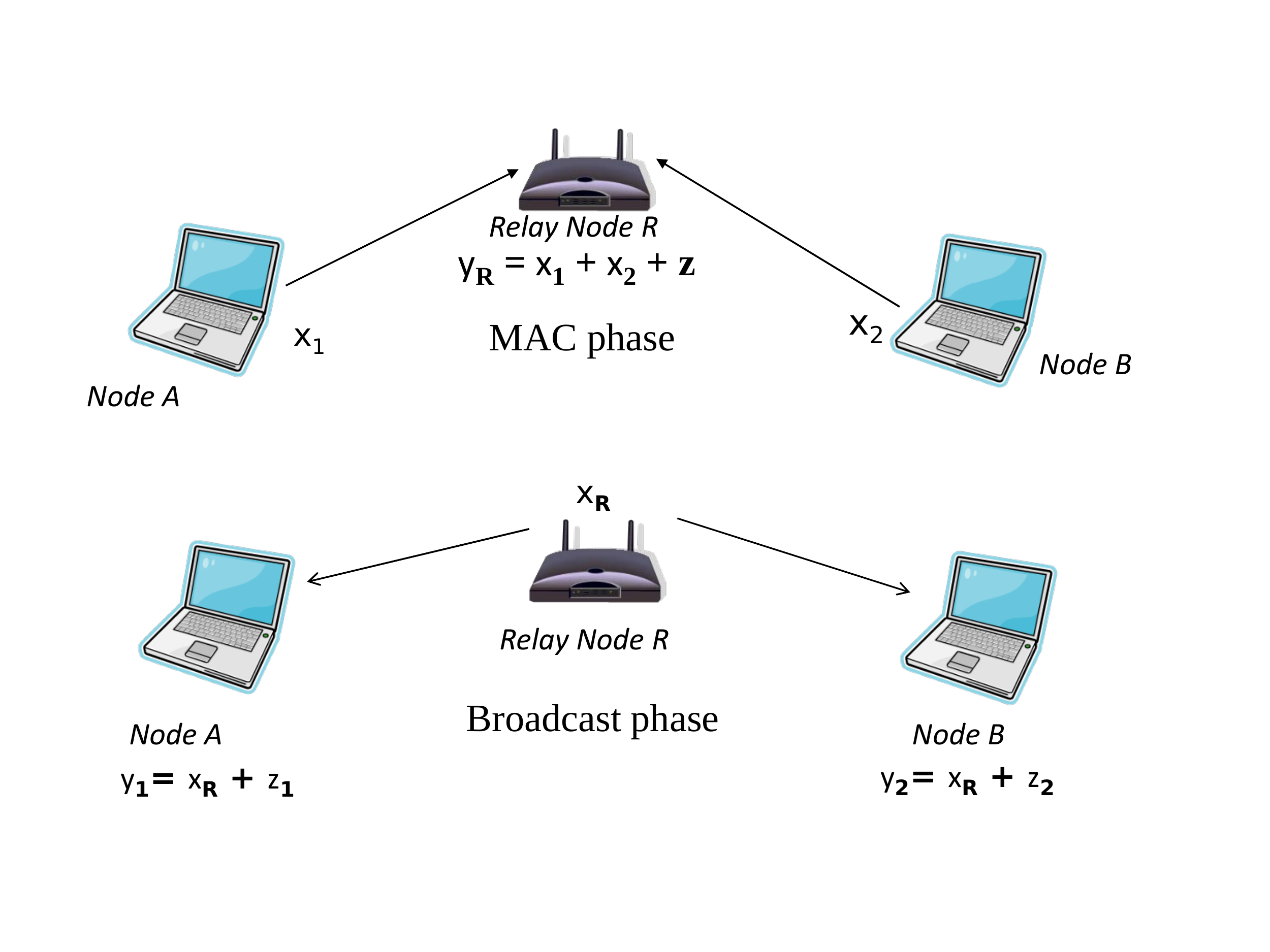}
  \caption{System Model with 3 Nodes}
  \label{fig:sysmodel}
\end{figure}

\paragraph{MAC phase} During the MAC phase, nodes $A$ and
$B$ transmit $\xv_1$ and $\xv_2$ in $n$ uses of an AWGN channel to
the relay. It is assumed that the two transmissions are
perfectly synchronized and, hence, the received signal at the relay
$\yv_R \in \mathbb{R}^n$ is given by
\[
\yv_R = \xv_1 + \xv_2 + \zv
\]
where the components of $\zv$ are independent identically
distributed (i.i.d) Gaussian random variables with zero mean and
variance $\sigma^2$. Further, it is assumed that there is an average
transmit power constraint of $P$ at both nodes and, hence, $E[||\xv_1||^2]
\leq nP$ and $E[||\xv_2||^2] \leq nP$.

\paragraph{Broadcast phase} During the broadcast phase, the relay node transmits
$\xv_R \in \mathbb{R}^n$ in $n$ uses of an AWGN broadcast channel to both nodes
$A$ and $B$. It is assumed that the average transmit power at the relay node is also
constrained to $P$ and that the noise variance at the two nodes is also $\sigma^2$.
Node $A$ forms an estimate of $\mathbf{u}_B$,  namely $\hat{\mathbf{u}}_B$ and
node $B$ forms an estimate of $\mathbf{u}_A$,  namely $\hat{\mathbf{u}}_A$. An error
is said to occur if either $\mathbf{u}_B \neq \hat{\mathbf{u}}_B$ or
$\mathbf{u}_A \neq \hat{\mathbf{u}}_A$, i.e., the probability of error is given by
\[
P_e \stackrel{\Delta}{=} \Pr \left(\{ \mathbf{u}_A \neq \hat{\mathbf{u}}_A\} \bigcup \{ \mathbf{u}_B \neq \hat{\mathbf{u}}_B \}\right)
\]

It is assumed that the communications in the MAC and broadcast
phases are orthogonal. For example, the communications in the MAC
and broadcast phase could be in two separate frequency bands (or in
two different time slots) and, hence, the MAC phase and broadcast
phase do not interfere with each other. To keep the discussion
simple, we will assume that the MAC and broadcast phases occur in
different time slots. This can be easily generalized to the case
when  $2n$ dimensions are available for communication, out of which
$n$ dimensions
 are allocated to the MAC phase and $n$ dimensions are allocated to the broadcast phase.
 Note that the signal-to-noise ratio (SNR) for all transmissions is ${\sf snr} = P/\sigma^2$ and, hence,
 we refer to it as simply the SNR without having to distinguish between the SNRs of the different phases.
Similarly, we restrict our attention to the case when the both nodes
$A$ and $B$ wish to exchange identical amount of information and,
hence, we can simply refer to one exchange rate without having to
distinguish between the rate for $A$ and $B$ separately.

Formally, we define the exchange rate $R_{ex,scheme}$ for an
encoding/decoding scheme as the maximum information rate that can be
exchanged reliably, i.e.,
\[
R_{ex,scheme} \stackrel{\triangle}= \max \frac{k}{n} : P_e
\rightarrow 0 \quad \mathrm{ as } \quad n \rightarrow \infty.
\]
The exchange capacity $C_{ex}$ is then the supremum of
$R_{ex,scheme}$ over all possible encoding/decoding schemes.


\section{Main Results and comments}
\label{sec:mainresults}
We mainly consider the case when the MAC and broadcast phase are both restricted to using exactly $n$ uses of the channel each. For this case, the main results in this paper are
\begin{itemize}

    \item The exchange capacity satisfies
    \[
    C_{ex} \leq \frac{1}{2} \log \left( 1+\frac{P}{\sigma^2} \right)
    \]
    because the MAC and broadcast phase each consist of $n$ AWGN channel uses.

    \item An exchange rate of
    \[
    R_{ex,Lattice} = \frac{1}{2} \log \left(\frac{1}{2}+\frac{P}{\sigma^2}\right)
    \]
    is achievable using the lattice coding scheme with dither and lattice decoding discussed in Section~\ref{sec:lattice}. The same exchange
    rate can also be obtained using the lattice coding scheme (without dither) and minimum angle decoding discussed in Section~\ref{sec:minangledecoding}. At high SNR, these lattice based coding and decoding schemes are nearly optimal because their rates
approach the upper bound.

    \item An exchange rate of
    \[
    R_{ex,JD}  =  \frac{1}{4} \log \left( 1 + \frac{2P}{\sigma^2} \right)
    \]
    is achievable using the joint decoding scheme in Section~\ref{sec:jointdecoding}. At low SNR, these scheme is nearly optimal because the rate approaches the upper bound.

    \item Clearly, any rate of the form
    $\beta R_{ex,JD} + (1-\beta) R_{ex,Lattice}$ is achievable
    for any $0 \leq \beta \leq 1$ by time sharing between the two schemes. This outperforms the recently proposed analog network coding idea in \cite{sachin07} over the entire range of SNR.

\end{itemize}

\section{Related Prior Work}
\label{sec:priorwork} Recently, there has been a significant amount
of work on coding for the bi-directional relay problem
\cite{larsson05}-\cite{feng07}. In \cite{sachin06}, Katti {\em et
al}., showed the usefulness of network coding for this problem.
Although they do not consider the physical layer explicitly in this
work,
 the natural extension of their solution to our problem would work as
follows. The $2n$ channel uses available for signaling would be
split into three slots with $2n/3$ channel uses each. In the first
time slot, $\uv_A$ is encoded using an optimal channel code for the
AWGN channel into $\xv_1$ and transmitted from node $A$. Similarly,
in the second time slot $\uv_B$ is encoded into $\xv_2$ and
transmitted from node $B$. At the relay, $\uv_A$ and $\uv_B$ are
decoded and then the relay forms $\uv_R = \uv_A \oplus \uv_B$ and
encodes $\uv_R$ into $\xv_R$ using an optimal code for the Gaussian
channel and broadcasts into both nodes. The two nodes decode $\uv_R$
and then since they have $\uv_A$ and $\uv_B$, they can obtain
$\uv_B$ and $\uv_A$ at the nodes $A$ and $B$, respectively. Here,
the physical layer and network layer are completely separated and
coding (or mixing of the information) is performed only at the
network layer. In the system model considered in
Fig.~\ref{fig:sysmodel}, the physical layer naturally performs
mixing of the signals from the two transmitters. The schemes that
take advantage of this can be referred to as joint physical layer
coding and network coding solutions.

 One such  scheme called analog network coding was recently proposed in \cite{sachin07}.
 In this case, the MAC phase and broadcast phase use $n$ uses of the AWGN and are orthogonal to each other.
 Gaussian code books are used at the transmitters to encode $\uv_A$ into $\xv_1$ and $\uv_B$ into $\xv_2$, respectively.
  Analog network coding is an amplify (rather than scale) and forward scheme where the received signal at the relay
   during the MAC phase $\yv_R$, is scaled to satisfy the power constraint and transmitted during the broadcast phase, i.e.,
     $\xv_R = \sqrt{\frac{P}{2P+\sigma^2}} \yv_R$.
     It can be seen that this scheme can achieve an exchange capacity of
$\frac{1}{2} \log \left(1+\frac{P/\sigma^2}{\frac{3P+\sigma^2}{P}} \right)$,
which is higher than that achievable with
the pure network coding scheme in \cite{sachin06} for high SNR.

 The schemes proposed in this paper can be thought of as decode and forward schemes which outperform
 the amplify and forward scheme in \cite{sachin07}.
In a very recent work \cite{popovski07}, it is conjectured that an
exchange rate of $\frac{1}{2}\log \left(1 +
\frac{P}{\sigma^2}\right)$ can be achieved, however, no scheme is
given or even conjectured. The scheme in this paper is a
constructive scheme that performs close to the rate conjectured in
\cite{popovski07}. The lattice decoding scheme discussed in
Section~\ref{sec:lattice} is similar to that used by Nazer and
Gastpar  \cite{nazerallerton07} for the problem of estimation the
sum of two Gaussian random variables, but was independently proposed
in the previous version of this paper \cite{narayanan07}.

\section{An Optimal Transmission Scheme for the BSC Channel}
\label{sec:bscchannel}

To motivate our proposed scheme, we first consider a system
where the physical layer channels are all binary symmetric channels.
i.e., $\xv_1 \in \{0,1\}^n$, $\xv_2 \in \{0,1\}^n$ and the signal
received at the relay is
\begin{equation}
\label{eqn:receivedbsc}
\yv_R = \xv_1 \oplus \xv_2 \oplus \ev
\end{equation}
where $\oplus$ denotes binary addition and
$\ev$ is an error sequence whose components are $0$ or $1$ with
probability $q$ and $1-q$ respectively and are i.i.d. Similarly,
during the broadcast phase also let the channel be a BSC channel
with crossover probability $q$.
In this case, an upper bound on the exchange capacity can be seen to be $1-H(q)$
since this is the maximum information that can flow to any of the nodes from the
relay. This can be achieved using the following coding scheme.

In this scheme, the two nodes use identical capacity achieving binary
{\em linear} codes of rate $1-H(q)$. Consider again the received
signal at the relay given in (\ref{eqn:receivedbsc}). Notice that
since $\xv_1$ and $\xv_2$ are codewords of the same linear code,
$\xv_1 \oplus \xv_2$ is also a valid codeword from the same code
which achieves capacity over a BSC channel with crossover
probability $q$. Hence, the relay can decode $\xv_1 \oplus \xv_2$
and transmit the result during the broadcast phase. The nodes $A$
and $B$ can also decode $\xv_1 \oplus \xv_2$ and since they have
$\xv_1$ and $\xv_2$, they can obtain $\xv_2$ and $\xv_1$,
respectively. This scheme achieves an exchange rate of $1-H(q)$ and
is therefore optimal.

\paragraph*{Random codes versus structured codes} It is quite interesting to note that
if random codes, i.e., codes from the Shannon ensemble were used instead of linear codes,
 $\xv_1 \oplus \xv_2$ cannot be decoded at the relay.
  The linearity (or group structure) of the code is exploited to make
  $\xv_1 \oplus \xv_2$ decodable at the relay and, hence,
   structured codes with a group structure outperform random codes for this problem.
Examples of schemes were structured codes outperform random codes have been given in
Korner and Marton \cite{korner} and more recently by Nazer and Gastpar in \cite{nazercompute07}, \cite{nazerstruc07}.

\section{Upper Bound on the Exchange Rate for Gaussian links}
Let us now return to the problem outlined in Section~\ref{sec:introduction}, where
the channels between the nodes and the relay are AWGN channels.
We restrict our attention to schemes where the MAC phase and broadcast phase
are both orthogonal to each other and use $n$ channel uses (or dimensions).
With this restriction, a simple upper bound on the
exchange rate can be obtained as follows. Consider a cut between the relay
node and node $A$. The maximum amount of information that can flow to either
 of the nodes from
the relay is $\frac12 \log \left(1 + \frac{P}{\sigma^2} \right)$.
Hence, the exchange
capacity is upper bounded by
\[
C_{ex,ub} = \frac12 \log \left(1 + \frac{P}{\sigma^2} \right).
\]

We now consider coding schemes and analyze their performance.

\section{Nested Lattice Based Coding Scheme with Lattice Decoding}
\label{sec:lattice}

 As shown in  Section~\ref{sec:bscchannel} for the BSC channel, codes with a group structure
(linear codes) enable decoding of a linear combination (or, sum)
codewords at the relay. This motivates the use of lattice codes for
the Gaussian channel since lattices have a similar group structure
with respect to real vector addition. We begin with some
preliminaries about lattices \cite{erez04}, \cite{erez05}.

An $n$-dimensional lattice $\Lambda$ is a subgroup of $\mathbb{R}^n$
under vector addition over the reals. This implies that if
$\lambda_1,\lambda_2 \in \Lambda$, then $\lambda_1 + \lambda_2 \in
\Lambda$. For any $\xv \in \mathbb{R}^n$, the quantization of $\xv$,
$Q_{\Lambda}(\xv)$ is defined as the $\lambda \in \Lambda$ that is
closest to $\xv$ with respect to Euclidean distance. The fundamental
Voronoi region $\mathcal{V}(\Lambda)$ is defined as
$\mathcal{V}(\Lambda) =\{\xv: Q_{\Lambda}(\xv) = \mathbf{0} \} $.
The mod operation is defined as $(\xv\mod\Lambda) = \xv -
Q_{\Lambda}(\xv)$. This can be interpreted as the error in
quantizing an $\xv$ to the closest point in the lattice $\Lambda$.
The second moment of a lattice is given by
\[
\sigma^2(\Lambda) = \frac{1}{V(\Lambda)} \frac{1}{n} \int_{{\cal
V}(\Lambda)} ||\xv||^2 \ d \xv,
\] and where $V(\Lambda)$, the volume of the fundamental Voronoi region is
denoted by $V(\Lambda) = \int_{{\cal V}(\Lambda)} d \xv$. The
normalized second moment of the lattice is then given by
\[
G(\Lambda) =
\sigma^2(\Lambda)/V(\Lambda)^{1/n}.
\]

Let us define the covering radius of a lattice $\mathcal{R}_u$ as
the radius of the smallest $n$-dimensional hyper sphere containing
the Voronoi region $\mathcal{V}(\Lambda)$. Also let $\mathcal{R}_l$
denote the effective radius of $\mathcal{V}(\Lambda)$, which is the
radius of the $n$-dimensional hyper sphere having the same volume as
$\mathcal{V}(\Lambda)$. Now we can define a Rogers-good Lattice
\cite[(69)]{erez04} as
\begin{equation}
1 \leq \left(\frac{\mathcal{R}_u}{\mathcal{R}_l}\right)^n < c \cdot
n \cdot (\log n)^a,
\end{equation}
where $c$ and $a$ are constants. This implies,
\begin{equation}\frac{1}{n} \log \frac{\mathcal{R}_u}{\mathcal{R}_l} =
\mathcal{O}\left(\frac{1}{n}\right).\end{equation}

 In this work, we are interested in
nested lattices. Formally, we can say the lattice $\Lambda_c$ (the
coarse lattice) is nested in the lattice $\Lambda_f$ (the fine
lattice) if $\Lambda_c \subseteq \Lambda_f$ \cite{zamir02}. Let the
fundamental Voronoi regions of the lattices, $\Lambda_c$ and
$\Lambda_f$  be $\mathcal{V}(\Lambda_c)$ and
$\mathcal{V}(\Lambda_f)$. The existence of nested lattices for which
$G(\Lambda_c) \approx 1/(2 \pi e)$ and $G(\Lambda_f) \approx 1/(2
\pi e)$ has been shown in \cite{zamir02}, \cite{erez05}. The number
of lattice points of the fine lattice in the basic Voronoi region of
the coarse lattice is given by the nesting ratio
$\left.\frac{V(\Lambda_c)}{V(\Lambda_f)}\right.$.

Lattice codes can be used to achieve capacity on the single user
AWGN channel under maximum likelihood (ML) \cite{buda89},
\cite{linder93}, \cite{urbanke98}. More recently, nested lattices
have been shown to achieve the capacity of the single user AWGN
channel under lattice decoding \cite{erez04}, \cite{erez05}. The
main idea in \cite{erez04}, \cite{erez05} is to use the coarse
lattice as a shaping region and the lattice points from the fine
lattice contained within the basic Voronoi region of the coarse
lattice as the codewords. The existence of good nested lattices that
achieve capacity has been shown in \cite{erez04}.

\subsection{Description}

We now describe our encoding and decoding schemes for the
bi-directional relaying problem using nested lattices. The encoding and decoding operations during the MAC
and broadcast phase are explained below. A general schematic is also
shown in Fig.~\ref{fig:lattice}.
\begin{figure}[t]\centering
  \includegraphics[width=4.0in]{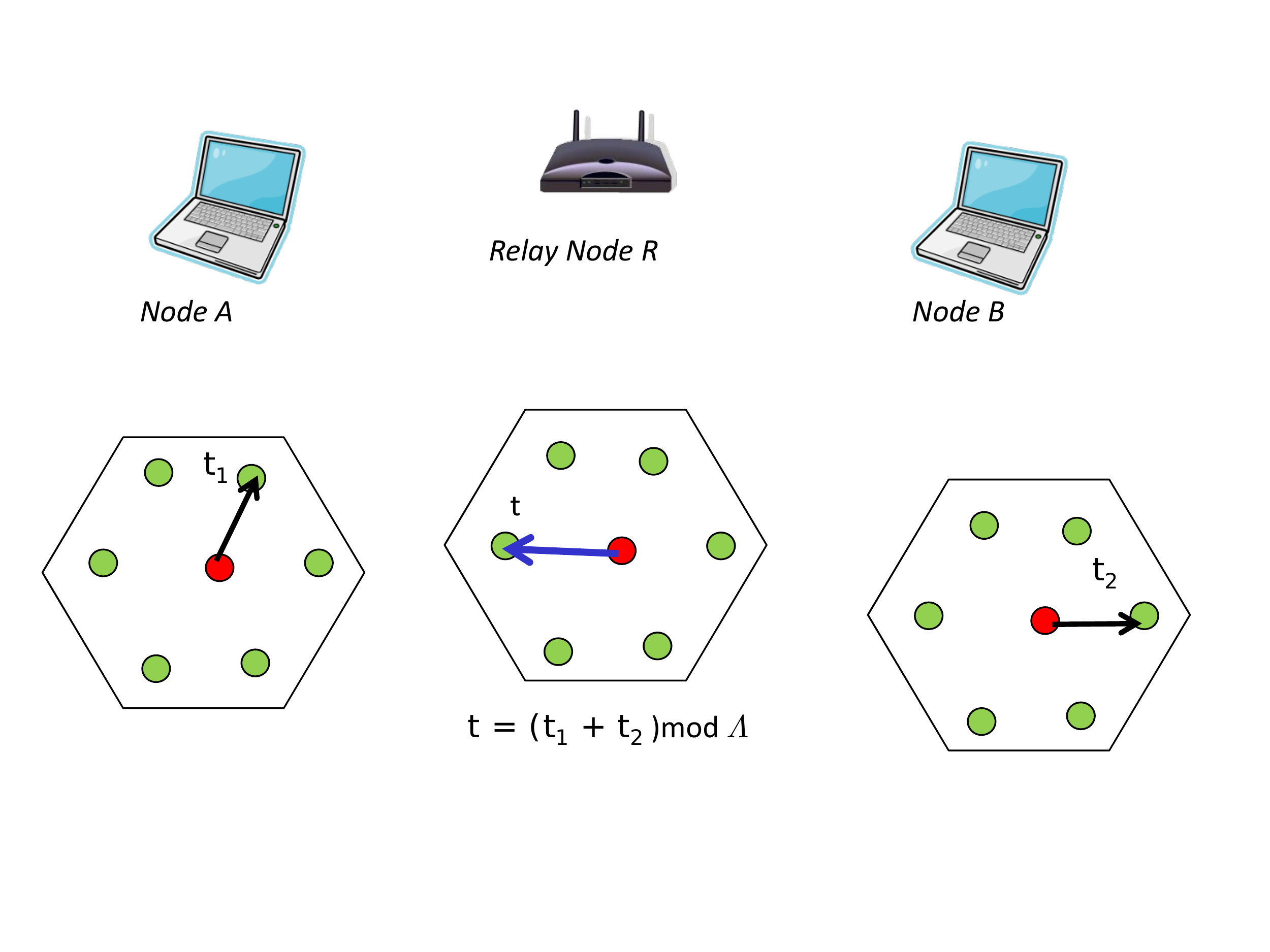}
  \caption{Lattice code based scheme showing the transmitted signals
  $\tv_1$, $\tv_2$ and the decoded signal at the relay $\tv$}\label{fig:lattice}
\end{figure}

\textit{MAC Phase:}
Let there be $k$ information bits in the information vector $\uv_A$
and $\uv_B$ and, hence, the exchange rate is $R = k/n$.
 At node $A$, the information vector $\uv_A$ is mapped onto a fine lattice point
  $\tv_1 \in \{\Lambda_f \cap \mathcal{V}(\Lambda_c)\}$, i.e., the set of all fine
   lattice points in the basic Voronoi region of the coarse lattice is taken to be the code.
   An identical code is used at node $B$ and the information vector $\uv_B$ is mapped onto
the codeword $\tv_2 \in \{\Lambda_f \cap \mathcal{V}(\Lambda_c)\}$. We then
generate dither vectors $\mathbf{d}_1$ and $\mathbf{d}_2$ which are
randomly generated
 $n$ dimensional vectors uniformly distributed over $\mathcal{V}(\Lambda_c)$.
  The dither vectors are mutually independent of each other and are
 known at both the relay node and the nodes $A$ and $B$.
Now node $A$ and node $B$ form the transmitted signal $\mathbf{x}_1$
and $\mathbf{x}_2$ as follows
\begin{equation} \mathbf{x}_1 = (\mathbf{t}_1-\mathbf{d}_1) \mod
\Lambda_c
\end{equation}
\begin{equation} \mathbf{x}_2 = (\mathbf{t}_2-\mathbf{d}_2) \mod
\Lambda_c
\end{equation}
By choosing an appropriate coarse lattice with second moment $P$,
the transmit power constraint will be satisfied at both nodes.
Assuming  perfect synchronization, the relay node receives $\yv_R$
given by
\begin{equation}
\label{eqn:receivedsignalmacphase}
\yv_R = \mathbf{x}_1+\mathbf{x}_2 + \zv
\end{equation} where $\zv$ is
the noise vector whose components have variance $\sigma^2$.

The main idea is that the relay decodes
\[
\mathbf{t} = (\mathbf{t}_1 + \mathbf{t}_2) \mod \Lambda_c
\]
from the received signal $\yv_R$.

\textit{Broadcast Phase:}
In the broadcast phase, the relay node transmits the index of
$\mathbf{t}$
  using a capacity achieving code for the AWGN channel.  The index
  (or, equivalently, $\mathbf{t}$)
  can be obtained at the nodes $A$ and $B$, and using a$\mod \Lambda$ operation,
  $\tv_2$ and $\tv_1$ can can be recovered at nodes $A$ and $B$, respectively.
This scheme can also be thought of as a decode and forward scheme
where the relay decodes a function of $\tv_1$ and $\tv_2$, namely
$\tv = (\tv_1 + \tv_2) \mod \Lambda_c$ and forwards this to the
nodes. Notice however, that the relay will not know either $\tv_1$
or $\tv_2$ exactly, it only knows $\tv$.

\subsection{Achievable rate}

Before we discuss the achievable exchange rate with the lattice encoding
scheme and decoding scheme discussed above, we need a few definitions and
lemmas.

The coding rate of the nested lattice code is defined as the
logarithm of the number of lattice points of the fine lattice in
$\Lambda_f \cap {\cal V}(\Lambda_c)$ which is given by
\begin{equation}
\nonumber \frac1n \log |\Lambda_f \cap {\cal V}(\Lambda_c) | =
\frac1n \log \frac{V(\Lambda_c)}{V(\Lambda_f)}
\end{equation}

\vspace{0.5cm}

\begin{lem}
\label{lemma:uniformdistribution} Let $t_1$ and $t_2$ be independent
random variables which are uniformly distributed over the set of all
fine lattice points in $\mathcal{V}(\Lambda_c)$, i.e.,
$\text{Pr}[t_1 = \lambda] = 1/2^{nR}, \ \forall \lambda \in
\{\Lambda_f \cap {\cal V}(\Lambda_c) \}$ and $\text{Pr}[t_2 =
\lambda] = 1/2^{nR}, \ \forall \lambda \in \{\Lambda_f \cap {\cal
V}(\Lambda_c) \}$, where $R$ is the coding rate of the nested
lattice code.  Then,  $t = (t_1 + t_2) \mod \Lambda$ is also
uniformly distributed over $\{\Lambda_f \cap {\cal V}(\Lambda_c)
\}$.
\end{lem}

\vspace{0.5cm}

\begin{proof}
\begin{eqnarray}
\text{Pr}[t = \lambda] &=& \sum_{i = 1}^{2^{nR}} \text{Pr}[t_1 =
\lambda_i]\text{Pr}[t = \lambda | t_1 = \lambda_i]\\
&=& \sum_{i = 1}^{2^{nR}} \frac{1}{2^{nR}}  \text{ Pr}[(\lambda_i + t_2)\mod \Lambda_c = \lambda | t_1 = \lambda_i]\\
&=& \sum_{i = 1}^{2^{nR}} \frac{1}{2^{nR}} \text{ Pr}[ t_2 = (\lambda -\lambda_i)\mod \Lambda_c | t_1 = \lambda_i]\\
&=& \sum_{i=1}^{2^{nR}}  \frac{1}{2^{nR}}  \frac{1}{2^{nR}} =
\frac{1}{2^{nR}}
\end{eqnarray}
\end{proof}

\vspace{0.5cm}

We restate some of definitions in \cite{erez04}. Let $\mathcal{R}_u$
be the covering radius of $\Lambda_c$ and let $\mathcal{B}(\mathcal{R}_u)$
be the $n$-dimensional ball of radius $\mathcal{R}_u$. Let $\rho^2$
be the second moment (per dimension) of the smallest ball containing
$\mathcal{V}(\Lambda_c)$, i.e.,
\begin{equation}
\rho^2 = \frac{1}{n} \frac{1}{\|\mathcal{B}(\mathcal{R}_u)\|}
\int_{\mathcal{B}(\mathcal{R}_u)} \|\uv \|^2 d\uv.
\end{equation}

Note that $\mathcal{V}(\Lambda_c)$ has a second moment $P$ and hence
$P < \rho^2$. Let $\mathbf{Z^*}$ be a random variable given by
$\mathbf{Z^*} = (1 - \alpha)(\mathbf{Z_1} + \mathbf{Z_2}) + \alpha
\mathbf{Z}$, with $\mathbf{Z_1} \sim \mathcal{N} ( 0 ,\rho^2
\mathbf{I}^n)$, $\mathbf{Z_2} \sim \mathcal{N} ( 0 ,\rho^2
\mathbf{I}^n)$ and $\mathbf{Z} \sim \mathcal{N} ( 0 ,\sigma^2
\mathbf{I}^n)$, where $\mathbf{I}^n$ is the $n \times n$ identity matrix.
The variance of $\mathbf{Z^*}$ satisfies

\vspace{0.5cm}

\begin{lem}[Modified version of {\cite[Lemma 6]{erez04}}]

\begin{equation}
\frac{n}{n+2} . \frac{2P\sigma^2}{2P + \sigma^2} \leq
\mathrm{Var(Z^*)} = (1 - \alpha)^2(\rho^2 + \rho^2) + \alpha^2
\sigma^2 \leq \left(\frac{R_u}{R_l}\right)^2 \frac{2P\sigma^2}{2P +
\sigma^2},
\end{equation}
 where $\mathcal{R}_u$ is the covering radius of the coarse lattice
$\Lambda_c$ and $\mathcal{R}_l$ is the radius of the $n$-dimensional hyper
sphere whose volume is equal to the volume of the basic Voronoi
region $\mathcal{V}(\Lambda_c)$.

\end{lem}

\vspace{0.5cm}

\begin{proof}
The proof of the above lemma closely follows the proof in \cite{erez04}. Equations
 \cite[116-118]{erez04} can be used to get a bound on Var$(Z_1)$ and
Var$(Z_2)$, i.e., $\rho^2$. Choosing $\alpha = \frac{2P}{2P + \sigma^2}$, the
modified version of \cite[Lemma 6]{erez04} can be proved.
\end{proof}

Let $\mathbf{X}_1$ and $\mathbf{X}_2$ be two independent
random variables which are uniformly distributed over ${\cal
V}(\Lambda_c)$ and let
$\mathbf{Z} \sim \mathcal{N} (
\mathbf{0} ,\sigma^2 \mathbf{I})$ be an $n$-dimensional Gaussian vector
independent of $\mathbf{X}_1$ and $\mathbf{X}_2$.
Further, let $\mathbf{Z}_{eq} =
(1-\alpha) (\mathbf{X}_1 + \mathbf{X}_2) + \alpha \mathbf{Z}$, where $\alpha =
\frac{2P}{2P+\sigma^2}$. Notice that $\mathbf{Z}_{eq}$ is not Gaussian.
We next state a Lemma which is a modified
version of  \cite[Lemma
11]{erez04} by Erez and Zamir which essentially shows that there exist
good lattices for which $\mathbf{Z}_{eq}$ can be well approximated by a Gaussian
of nearly the same variance and the approximation gets better as $n \rightarrow \infty$.

\vspace{0.5cm}

\begin{lem}[Modified version of {\cite[Lemma 11]{erez04}}]
 Let $\Lambda_c$ be a lattice which is both Rogers-good
as well as Poltyrev-good. Then, for any $\xv$,
\begin{equation}f_{Z_{eq}}(\xv) < e^{\epsilon_1(\Lambda_c)n}
f_{\mathbf{Z^*}}(\xv),\end{equation} where
$\epsilon_1(\Lambda_c) \rightarrow 0$ as $n \rightarrow \infty$.

\end{lem}

\vspace{0.5cm}

\begin{proof}
 To prove the modified version of \cite[Lemma 11]{erez04},
  we can repeat the steps in \cite{erez04}, and equation
\cite[200]{erez04} can be restated with the notation in this paper
as
\begin{equation}\label{eqn:lem111} \frac{1}{n} \log
\frac{f_{\mathbf{x_1}}(\mathbf{x})}{f_{\mathbf{Z_1}}(\xv)} =
\epsilon_1(\Lambda_c) = \mathcal{O} \left(\frac{1}{n}\right).\end{equation}
Similarly for $\mathbf{Z}_2$, we get
\begin{equation}\label{eqn:lem112}
\frac{1}{n} \log \frac{f_{\mathbf{X_2}}(\xv)}{f_{\mathbf{Z_2}}(\xv)}
= \epsilon_1(\Lambda_c) = \mathcal{O}
\left(\frac{1}{n}\right).\end{equation} Combining (\ref{eqn:lem111})
and (\ref{eqn:lem112}) and also the definition of $\mathbf{Z^*}$ as
$\mathbf{Z^*} = (1 - \alpha)(\mathbf{Z_1} + \mathbf{Z_2}) + \alpha
\mathbf{N}$, we can get the proof of the modified version of
\cite[Lemma 11]{erez04}.
\end{proof}

\vspace{0.5cm}

We next state the
theorem which is an application of the above Lemmas. This is
very similar to \cite[Theorem 5]{erez04}.
\vspace{0.5cm}

\begin{thm}[modified version of {\cite[Theorem 5]{erez04}}]
\label{NestedLatticeTheorem}
Let $\mathbf{x}_1$ and $\xv_2$ be realizations of two independent
random variables which are uniformly distributed over ${\cal
V}(\Lambda_c)$ and let $\mathbf{z}$ be a realization of an
$n$-dimensional Gaussian vector $\mathbf{Z} \sim \mathcal{N} (
\mathbf{0} ,\sigma^2 \mathbf{I})$. Further, let $\zv_{eq} =
(1-\alpha) (\xv_1 + \xv_2) + \alpha \zv$, where $\alpha =
\frac{2P}{2P+\sigma^2}$. For any coding rate $R < \frac12
\log(\frac12 + \frac{P}{\sigma^2})$, there exists a sequence of
 $n$-dimensional nested lattice pairs
$(\Lambda_f^{(n)}, \Lambda_c^{(n)})$ of coding rate $R$ for which
\[
\Pr \{ \mathbf{Z}_{eq} \not\in {\cal{V}}_{f}^{(n)} \} \rightarrow 0,
\ \mbox{as}, \  n\rightarrow \infty.
\]
\end{thm}

\vspace{0.5cm}
\begin{proof}
The proof follows closely the proof of  \cite[Theorem 5]{erez04}. We
mention the places where the proof in \cite{erez04} that have to be
modified. Equation (81) in \cite{erez04} must be modified to take
into account that we have two transmitters. Equation
\cite[(81)]{erez04} must be modified with Lemma 2. Also Equation
\cite[(82)]{erez04} must be modified with Lemma 3. After this, we
can continue with the proof in \cite{erez04} by calculating the
Poltyrev exponents and also using the fact that Rogers-good and
Poltyrev-good lattices exist. Continuing with these steps in
\cite{erez04} shows that we can obtain the rate of $R<\frac12
\log\left(\frac12 + \frac{P}{\sigma^2}\right)$, which proves the
theorem.
\end{proof}

\vspace{0.5cm}

 We are ready to present the main theorem in this
Section which is given below.
\vspace{0.5cm}

\begin{thm}
For the lattice coding scheme described above, any exchange rate $R_{ex,lattice} < \frac{1}{2}\log(\frac{1}{2} +
\frac{P}{\sigma^2})$ is achievable with lattice (Euclidean) decoding.
\end{thm}

\vspace{0.5cm}

\begin{proof}

\underline{MAC Phase:} We will first show that the probability of error
in decoding $\mathbf{t}$ from $\yv_R$ can be made arbitrarily small for
asymptotically large $n$. During the MAC phase, the relay tries to decode to
\[
\mathbf{t} = (\mathbf{t}_1 + \mathbf{t}_2) \mod \Lambda
\]
from the received signal $\yv_R$ (given in
(\ref{eqn:receivedsignalmacphase}) as follows. The decoder at the
relay node forms $\mathbf{\hat{t}} = (\alpha \yv_R + \mathbf{d}_1 +
\mathbf{d}_2) \mod \Lambda_c$ ($\alpha$ will be determined later)
and finds the lattice point in the fine lattice that is closest to
$\mathbf{\hat{t}}$, i.e., the estimate of $\tv$ is
$Q_{\Lambda_f}(\mathbf{\hat{t}})$. Using the distributive property
of the mod operation, $\mathbf{\hat{t}}$ can be written as:

\begin{eqnarray}
\nonumber
 \mathbf{\hat{t}}&=& (\alpha \yv + \dv_1 + \dv_2) \mod
 \Lambda_c \\
 \nonumber
               &=& (\alpha (\xv_1 + \xv_2) +\alpha \zv  + \dv_1 + \dv_2) \mod
 \Lambda_c \\
 \nonumber
 &=& ( \mathbf{x}_1 +\mathbf{x}_2 + \mathbf{d}_1 + \mathbf{d}_2 + \alpha \zv -(1-\alpha) (\mathbf{x}_1 +\mathbf{x}_2) ) \mod
 \Lambda_c\\
 \nonumber
 &=& ( (\mathbf{t}_1-\mathbf{d}_1) \mod \Lambda_c +(\mathbf{t}_2-\mathbf{d}_2) \mod \Lambda_c + \\
 \nonumber
 &&\mathbf{d}_1 + \mathbf{d}_2+\alpha \zv -(1-\alpha) (\mathbf{x}_1 +\mathbf{x}_2) ) \mod
 \Lambda_c\\
 \nonumber
 &=& ((\mathbf{t}_1 + \mathbf{t}_2) \mod \Lambda_c +\alpha \zv -(1-\alpha)
 (\mathbf{x}_1
 +\mathbf{x}_2)) \mod \Lambda_c\\
 &=& ( \mathbf{t} +\alpha \zv -(1-\alpha) (\mathbf{x}_1
 +\mathbf{x}_2)) \mod \Lambda_c
\end{eqnarray}

Due to the group structure of the lattice, $\mathbf{t}$ is a lattice
point in the fine lattice (more precisely, $\mathbf{t} \in
\{\Lambda_f \cap {\cal V}(\Lambda_c) \}$). From Lemma 1, it can be
seen that $\mathbf{t}$ is uniformly distributed over $\mathcal{
V}(\Lambda_c) $. Further, note that $\tv_1$ and $\tv_2$ are
independent of $\zv$, $\xv_1$ and $\xv_2$ and, hence,  we can define
an equivalent noise term as $\zv_{eq} = \alpha \zv - (1 -
\alpha)(\xv_1 + \xv_2) $ such that $\mathbf{t}$ and $\zv_{eq}$ are
independent of each other. The second moment of $Z_{eq}$ is given by
$\sigma_{eq}^2 = E[Z_{eq}^2] = \alpha^2 \sigma^2 + (1-\alpha)^2 2P$.
We can now choose $\alpha$ to minimize $E[z_{eq}^2]$ and the
resulting optimum values of $\alpha$ and $\sigma_{eq}^2$ are
$\alpha_{opt} = \frac{2P}{2P+\sigma^2}$ and $\sigma_{eq,opt}^2 =
\frac{2P\sigma^2}{2P+\sigma^2}$. From Theorem~\ref{NestedLatticeTheorem},
it can be seen that
there exist nested lattices of rate $R_{lattice} <\frac{1}{2} \log
\left(\frac12+\frac{P}{\sigma^2}\right)$ for which $Pr\{
\mathbf{Z}_{eq} \not\in {\cal{V}}_{f}^{(n)} \} \rightarrow 0, \
\mbox{as}, \ n\rightarrow \infty. $ and, hence, the probability of
decoding error
\[
P_e = Pr\{Q_{\Lambda_f}(\mathbf{\hat{t}}) \neq \mathbf{t} \} \rightarrow 0
\ {\mbox {as}} \ n \rightarrow \infty.
\]
Hence, we can use a rate of
\begin{equation}
R_{Lattice} \stackrel{\triangle}= \frac{1}{n} \log_2 |\Lambda_f \cap
{\cal V}(\Lambda_c)| < \frac{1}{2} \log
\left(\frac{1}{2}+\frac{P}{\sigma^2}\right)
\label{eqn:Rlattice}
\end{equation}
at each of the nodes and $(\tv_1 + \tv_2) \mod \Lambda_c$ can be
decoded at the relay.

\underline{Broadcast Phase:}

In the broadcast phase, the relay node transmits the index of $\mathbf{t}$
  using a capacity achieving code for the AWGN channel. Since the capacity
  of the AWGN is $\frac{1}{2} \log \left(1 + \frac{P}{\sigma^2} \right)$ which
  is higher than $R_{Lattice}$ in (\ref{eqn:Rlattice}) and, hence, the index
  (or, equivalently, $\mathbf{t}$)
  can be obtained at the nodes $A$ and $B$. Since node $A$ already has $\uv_A$
  and, hence, $\tv_1$, it needs to recover $\uv_B$ or, equivalently, $\tv_2$, from
  $\tv$ and $\tv_1$. This can be done as follows. Node $A$ computes
  $(\tv - \tv_1) \mod \Lambda_c$, which can be written as

\begin{eqnarray}
\nonumber
 (\mathbf{t} - \mathbf{t_1}) \mod \Lambda_c &=& ((\mathbf{t}_1 + \mathbf{t}_2) \mod \Lambda_c - \mathbf{t}_1 ) \mod \Lambda_c\\
 \nonumber
 &=& \mathbf{t}_2 \mod \Lambda_c\\
 &=& \mathbf{t}_2
 \end{eqnarray}
Similarly, $\mathbf{t}_1$ can also be obtained in node $B$ by taking
$ (\mathbf{t} - \mathbf{t}_2) \mod \Lambda_c$.
Hence, an effective rate of $R_{ex,Lattice} < \frac{1}{2} \log
\left(\frac{1}{2}+\frac{P}{\sigma^2}\right)$ can be obtained using
nested lattices with lattice decoding.
\end{proof}

\vspace{0.5cm}

We conclude by noting that, at high SNR,  this scheme approaches the upper
bound of $\frac{1}{2} \log \left(1+\frac{P}{\sigma^2}\right)$ and is
therefore nearly optimal.

This scheme can be interpreted as a Slepian-Wolf coding scheme using
nested lattices, i.e., the relay wishes to convey $\tv_1 + \tv_2$ to
node $A$, where some side information (namely, $\tv_1$) is
available. Thus, the broadcast phase in effect uses nested lattices
for solving the Slepian-Wolf coding problem.
This scheme can also be thought of as a decode and forward scheme
where the relay decodes a function of $\tv_1$ and $\tv_2$, namely
$\tv = (\tv_1 + \tv_2) \mod \Lambda_c$ and forwards this to the
nodes. Notice however, that the relay will not know either $\tv_1$
or $\tv_2$ exactly, it only knows $\tv$.

Since the nested lattice code we have used is a capacity achieving code for the AWGN channel,
 one does not have to encode $\tv$
again. The relay can simply transmit  $\tv$ to the nodes $A$ and $B$
and $\tv$ can be decoded at the nodes $A$ and $B$ in the presence of
noise at the nodes $A$ and $B$.
 Notice that $E[||\tv||^2] \leq P$ and, hence,
the power constraint will be satisfied at the relay node.

\section{Lattice Coding with Minimum Angle Decoding} \label{sec:minangledecoding}

 In the previous section we observed that nested lattice decoding
 can achieve an exchange rate of $\frac{1}{2} \log(\frac{1}{2} + \frac{P}{\sigma^2})$ with
 lattice decoding alone.
 This naturally leads us to the question, can we achieve a better
 performance by using other decoding schemes? In this section we study a suboptimal decoder
 called the minimum angle decoder \cite{urbanke98}.
\subsection{Description}

 We next briefly explain our minimum angle decoding scheme.
       We have two transmitters communicating simultaneously to a central router.
       Both of them have the same power constraint $P$. The noise in the channel
       is Gaussian having a variance $\sigma^2$. As we have seen in previous sections, choosing
        a good lattice code provides us with a considerable rate gain compared to random codes.
        Here we consider a $n$-dimensional lattice $\Lambda_n \subset
       \mathds{R}^n$. Let $T_{\sqrt{P}}$ be an $n$-dimensional closed ball,
       centered at the origin and having a radius $\sqrt{nP}$, and let
       $V_n(\sqrt{nP})$ be the hyper-volume of $T_{\sqrt{P}}$.
       This can be treated as a power constraint. Our
       codewords will be composed of lattice points in the sphere
       $T_{\sqrt{P}}$. Our encoding strategy would be that, each transmitter
       chooses a lattice point corresponding to its message index
       and transmits synchronously over the the Gaussian channel.
       Here we have no nested lattice construction or the use of
       dither in the encoding stage.
       At the receiver we will be interested in decoding to the sum
       of these lattice points.

\textit{Minimum Angle Decoder:} A minimum angle decoder discussed
here makes a decision based on lattice points in a thin
n-dimensional spherical shell $T_{\sqrt{2P}}^\Delta \doteq \{\xv \in
\mathds{R}^n : \sqrt{n(2P-\delta)} \leq \|\xv\| \leq
\sqrt{n(2P+\delta)} \}$, $\delta$ is small and non-zero.  It takes
the received vector and finds the lattice point, whose projection on
the thin shell, is closest to the received vector.

An optimal decoder can be seen to be
  a maximum likelihood decoder. However, it is
  very tough to analyze this decoder. Hence, we analyze the minimum angle decoder,
   which is more tractable analytically.
We next state the main theorem of this section,

\subsection{Achievable rate}

\begin{thm}
For the bi-directional relaying problem considered in Section~\ref{sec:introduction},
there exists at least one $n$-dimensional
lattice $\Lambda_n$ such that an exchange rate of $R_{ex} < \frac{1}{2}
\log(\frac{1}{2} + SNR)$ is achievable using a minimum angle decoder as
$n \rightarrow \infty$.
\end{thm}

\subsection*{Proof Sketch}

\begin{figure}[htbp]
\begin{center}
\includegraphics[width=4in,angle = 0]{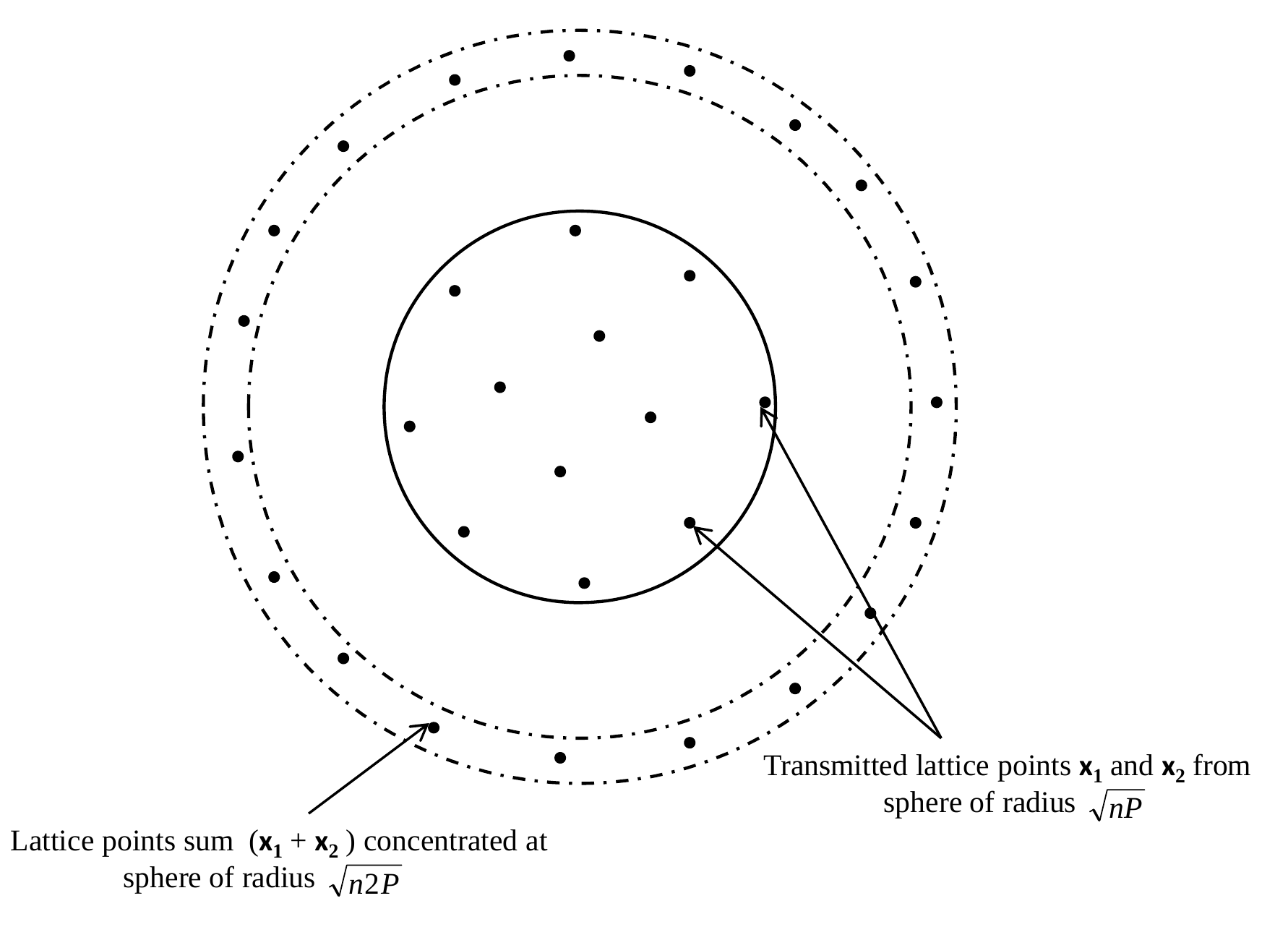}
\end{center}
\caption{Picture showing the concentration of the sum of lattice
points on the thin shell of radius $\sqrt{n2P}$. }
\label{fig:sphereconcentration}\end{figure}

          It is well known that the volume of an $n$-dimensional sphere
          is concentrated mainly on the surface of the sphere as
          the dimension becomes large. It is also known that, if we
          intersect a lattice with a $n$-dimensional sphere, then most of
          the lattice points will be concentrated very close to the
          surface \cite{urbanke98}.  In the course of our proof, we will show
          that the sum of any two such randomly chosen lattice
          points is also concentrated on a thin spherical shell $T_{\sqrt{2P}}^\Delta$
          at a radius of $\sqrt{2nP}$.
          Hence, the probability of error will be
          largely depended on the lattice points in the thin spherical shell $T_{\sqrt{2P}}^{\Delta}$.
          We will be using the Blichfeldt's principle to
          show that there exist translations (one for each user) of the lattice $\Lambda_n$
          where the sum points are concentrated (see Theorem~\ref{thm:blichfeldt}, Lemma~\ref{lem:Hyper Volume Concentration Lemma}
           and Lemma~\ref{lem:blichtranslationrate}).
          Once concentration for the sum of lattice points can be established,
           we can perform minimum angle
          decoding.  In minimum angle decoding, we are interested only
          in the angle between the different lattice points on the thin
          spherical shell. It must be noted that the choice of the lattice
          $\Lambda_n$ must be such that it must act as a good channel code. The
          Minkowski-Hlawka theorem (Theorem~\ref{thm:Minkowski-hlawka} and Lemma~\ref{lem:existencelatticeMinkowski})
         can be used to show existence of
          such lattices. Choosing volume of the lattice's Voronoi region appropriately
          allows us to compute an achievable rate of this scheme.


\subsection{Detailed proof}

     First let us provide some definitions and notation,
    \begin{itemize}
    \item $\Lambda_n$ denotes an $n$-dimensional lattice and let $P_n$
     be the basic Voronoi region of the lattice.
    \item  $\mathbf{s_1},\mathbf{s_2}$ are two $n$-dimensional translational
    vectors,
       $\mathbf{s_1},\mathbf{s_2} \in P_n$.
      \item $\mathcal{L}(\mathbf{s},T) $ is defined as the set of  lattice points translated by
      a vector $\sv$ present in the region $T$ or otherwise, $\mathcal{L}(\mathbf{s},T)=(\Lambda_n + \sv) \cap T
      $.
      \item $T_{\sqrt{P}}$ and $T_{\sqrt{2P}}$ represents
      $n$-dimensional hyper-spheres centered at the origin and having
      a radius of $\sqrt{nP}$ and $\sqrt{n2P}$ respectively.
      \item $T_{\sqrt{2P}}^\Delta$ represents $n$-dimensional thin spherical shell around the radius
      $\sqrt{n2P}$ and is defined as $T_{\sqrt{2P}}^\Delta = \{\xv \in
\mathds{R}^n : \sqrt{n(2P-\delta)} \leq \|\xv\| \leq
\sqrt{n(2P+\delta)} \}$, $\delta$ is small and non-zero.
\item $\mathcal{C}_1$ and $\mathcal{C}_2$ are codewords formed by
intersection of lattice points of hyper-spheres and are given by,
$\mathcal{C}_1 = \mathcal{L}(\mathbf{s_1},T_{\sqrt{P}})=(\Lambda_n +
\sv_1) \cap T_{\sqrt{P}} $ and $\mathcal{C}_2 =
\mathcal{L}(\mathbf{s_2},T_{\sqrt{P}})=(\Lambda_n + \sv_2) \cap
T_{\sqrt{P}} $.
\item
Similarly $\mathcal{C}_{\sqrt{2P}}$ is given as
$\mathcal{C}_{\sqrt{2P}} = \mathcal{L}(\mathbf{s_1} +
\mathbf{s_2},T_{\sqrt{2P}})=(\Lambda_n + \sv_1 + \sv_2) \cap
T_{\sqrt{2P}} $ and $\mathcal{C}_{\sqrt{2P}}^{\Delta}$ is given by
$\mathcal{C}_{\sqrt{2P}}^{\Delta} = \mathcal{L}(\mathbf{s_1} +
\mathbf{s_2},T_{\sqrt{2P}}^{\Delta})=(\Lambda_n + \sv_1 + \sv_2)
\cap T_{\sqrt{2P}}^{\Delta} $.
 \item $\mathcal{C}_2$ are codewords
formed by intersection of lattice points of hyper-spheres and are
given by, $\mathcal{C}_1 =
\mathcal{L}(\mathbf{s_1},T_{\sqrt{P}})=(\Lambda_n + \sv_1) \cap
T_{\sqrt{P}} $ and $\mathcal{C}_2 =
\mathcal{L}(\mathbf{s_2},T_{\sqrt{P}})=(\Lambda_n + \sv_2) \cap
T_{\sqrt{P}} $.
\item Let $\mathcal{C}_{\oplus}$ be defined as
     $\mathcal{C}_{\oplus} = \mathcal{C}_1 \times \mathcal{C}_2 = \{(\mathbf{x_1},\mathbf{x_2}) : \mathbf{x_1} \in \mathcal{C}_1, \mathbf{x_2} \in \mathcal{C}_2
      \}.$
      This denotes the combined collection of pairs of codewords of both the
      transmitters.
      \item
      Again let $\mathcal{C}_{\oplus}^{\Delta} =
      \{(\mathbf{x_1},\mathbf{x_2}) :
       \mathbf{x_1} + \mathbf{x_2} \in T_{\sqrt{2P}}^{\Delta} ,
       \mathbf{x_1} \in \mathcal{C}_1, \mathbf{x_2} \in \mathcal{C}_2
       \}$ This denotes the codeword pairs whose sum lies on the
       thin shell $T_{\sqrt{2P}}^{\Delta}$.
       \item Let $\mathcal{C}_{\oplus}' = \mathcal{C}_{\oplus} \setminus \mathcal{C}_{\oplus}^{\Delta}
       $. This denotes the code word pairs whose sum does not lie
       on the thin shell. It must be noted that the set formed by the sum of
       codewords in $\mathcal{C}_{\oplus}^{\Delta}$ need not be the same as
       $\mathcal{C}_{\sqrt{2P}}^{\Delta}$ and at low SNR this may
       lead to significant difference between ML and minimum angle
       decoding.
     \item Let $M_1$, $M_2$, $M_{\oplus}$, $M_{\oplus}^{\Delta}$ and $M_{\oplus}'$ denote the cardinality of $\mathcal{C}_1$,
      $\mathcal{C}_2$, $\mathcal{C}_{\oplus}$, $\mathcal{C}_{\oplus}^{\Delta}$ and $\mathcal{C}_{\oplus}'$  respectively.

\item
For a given code $\mathcal{C}$, let us denote the average
probability of error, under minimum distance decoding as
$P^{\mathcal{C}}$.

\item
 Let us define a projection function $\pi$. This projects a $n$
 dimensional vector onto to an inner sphere of radius
 $\sqrt{n(2P-\delta)}$. It is defined as $\pi(\xv) = (\sqrt{n(2P-\delta)}/\|\xv \|)\xv$.

\end{itemize}

 It is easy to see that minimum distance decoding
is equivalent to maximum likelihood decoding in the presence of
Gaussian noise.

As mentioned before the set of lattice points whose sum lies in the
thin spherical shell $T_{\sqrt{2P}}^{\Delta}$ is much larger than
the lattice points whose sum lies outside the spherical region.
Hence the average probability of error will not be affected much by
these lattice points. This motivates us to express the average
probability of error $P^{\mathcal{C}_{\oplus}}$ as a sum of two terms
as made more clear in the lemma below\cite{urbanke98}.

\begin{lem}

\begin{equation*}P^{\mathcal{C}_{\oplus}} \leq \frac{M_{\oplus}'}{M_{\oplus}} +
\frac{1}{M_\oplus} \sum_{(\mathbf{x_1},\mathbf{x_2}) \in
\mathcal{C}_{\oplus}^{\Delta}}{P}^{\mathcal{C}_{\oplus}^{\Delta}}(\pi(\mathbf{x_1}
+ \mathbf{x_2}))
\end{equation*}

 \end{lem}

\begin{proof}

Let the ordered pair $(\mathbf{x_1},\mathbf{x_2}) \in
\mathcal{C}_{\oplus}$, denote that $\mathbf{x_1} \in \mathcal{C}_1$
and $\mathbf{x_2} \in \mathcal{C}_2$. We next follow similar steps
of the proof as given in \cite{urbanke98}. Let
${P}^{\mathcal{C}_{\oplus}}(\mathbf{x_1},\mathbf{x_2}) =
{P}^{\mathcal{C}_{\oplus}}(\mathbf{x_1}+ \mathbf{x_2}) $, denote the
probability of error in decoding to the sum $(\mathbf{x_1}+
\mathbf{x_2})$, for a pair  $ (\mathbf{x_1},\mathbf{x_2}) \in
\mathcal{C}_{\oplus}$. $\mathbf{x_1}$ and $ \mathbf{x_2}$ are
transmitted simultaneously by user $1$ and $2$ respectively and the
router receives the sum $\mathbf{x_1} + \mathbf{x_2}$ corrupted with
some Gaussian noise. Now let $\tilde{P}^{\mathcal{C}_{\oplus}}$,
represent a suboptimum decoder which maps the received point to the
nearest sum $\mathbf{x_1} + \mathbf{x_2}$, where
$(\mathbf{x_1},\mathbf{x_2}) \in \mathcal{C}_{\oplus}^{\Delta}$.
Hence we can bound the average probability of error as follows

\begin{eqnarray*}
 {P}^{\mathcal{C}_{\oplus}} &\leq&
 \tilde{P}^{\mathcal{C}_{\oplus}}\\
 &=& \frac{1}{M_\oplus} \sum_{(\mathbf{x_1},\mathbf{x_2})
\in
\mathcal{C}_{\oplus}'}\tilde{P}^{\mathcal{C}_{\oplus}}(\mathbf{x_1},\mathbf{x_2})
+ \frac{1}{M_\oplus} \sum_{(\mathbf{x_1},\mathbf{x_2}) \in
\mathcal{C}_{\oplus}^{\Delta}}\tilde{P}^{\mathcal{C}_{\oplus}}(\mathbf{x_1},\mathbf{x_2})\\
&\stackrel{(a)}=& \frac{M_{\oplus}'}{M_{\oplus}} +
\frac{1}{M_\oplus} \sum_{(\mathbf{x_1},\mathbf{x_2}) \in
\mathcal{C}_{\oplus}^{\Delta}}{\tilde{P}}^{\mathcal{C}_{\oplus}}(\mathbf{x_1},\mathbf{x_2})\\
&=& \frac{M_{\oplus}'}{M_{\oplus}} + \frac{1}{M_\oplus}
\sum_{(\mathbf{x_1},\mathbf{x_2}) \in
\mathcal{C}_{\oplus}^{\Delta}}{\tilde{P}}^{\mathcal{C}_{\oplus}}(\mathbf{x_1}+\mathbf{x_2})\\
&\stackrel{(b)}\leq& \frac{M_{\oplus}'}{M_{\oplus}} +
\frac{1}{M_\oplus} \sum_{(\mathbf{x_1},\mathbf{x_2}) \in
\mathcal{C}_{\oplus}^{\Delta}}{P}^{\pi(\mathcal{C}_{\oplus}^{\Delta})}(\mathbf{x_1}+\mathbf{x_2})\\
\end{eqnarray*}

Here in $(a)$ the first term follows since $M_{\oplus}'$ is the
cardinality of $\mathcal{C}_\oplus'$ and the probability of error
using our suboptimal decoder is 1, when $(\mathbf{x_1},\mathbf{x_2})
\in \mathcal{C}_{\oplus}'$. In (b) $\pi$ refers to the projection
function defined before and
${P}^{\pi(\mathcal{C}_{\oplus}^{\Delta})} $, represents the minimum
angle decoder described before. The inequality can be shown to be
true by referring to the discussions on  \cite[Lemma 2]{urbanke98}.

 \end{proof}

Next since we are interested in lattice points projected on to the
inner sphere of radius $\sqrt{n(2P-\delta)}$, we can define a
decoding algorithm that looks at the angle between the lattice
points, the minimum angle decoder. We next establish some more
definitions. Let $B_\theta(\yv)$ denote a n-dimensional cone
centered at the origin and having the axis passing through $\yv$.
Let $\theta$ be the half-angle of the cone and $\yv \in
\mathds{R}^n$ be non-zero.

 We next define a sub-optimum decoding function  given as follows,
 \begin{equation*} A_{\theta}(\mathbf{x_1},\mathbf{x_2}) = B_{\theta}(\mathbf{x_1}+\mathbf{x_2}) \setminus
 \bigcup_{\mathbf{x'} \in \mathcal{C}_{\sqrt{2P}}^{\Delta}\setminus \{\mathbf{x_1}+\mathbf{x_2}\}} B_{\theta}(\mathbf{x'})
 \end{equation*} or this can also be expressed as
\begin{equation*} A_{\theta}^{C}(\mathbf{x_1},\mathbf{x_2}) = B_{\theta}^{C}(\mathbf{x_1}+\mathbf{x_2})
 \bigcup_{\mathbf{x'} \in \mathcal{C}_{\sqrt{2P}}^{\Delta}\setminus \{\mathbf{x_1}+\mathbf{x_2}\}} B_{\theta}(\mathbf{x'})
 \end{equation*}
$A_\theta(\mathbf{x_1},\mathbf{x_2})$ represents the region of the
cone $B_{\theta}(\mathbf{x_1}+\mathbf{x_2})$, that does not
intersect with any other cone corresponding to the other lattice
codeword points $\mathbf{x'}$, located in the thin spherical shell.
During decoding, when we receive a vector that falls in the region
$A_\theta(\mathbf{x_1},\mathbf{x_2})$, we decode to the sum codeword
$(\mathbf{x_1} + \mathbf{x_2})$. It may not be possible to decode to
the individual codewords $\mathbf{x_1}$ and $\mathbf{x_2}$, as
different pairs of codewords may yield the same sum. However it must
be noted, that in the forward phase, we are interested in decoding
only to the sum of the transmitted codewords.

        Let $P_{\theta}$ denote the probability of error using the
        sub-optimum decoder. Then, we have
\begin{eqnarray*}
P_{\theta}^{\pi(\mathcal{C}_{\oplus}^{\Delta})}
(\mathbf{x_1},\mathbf{x_2}) &=& \mbox{Pr}\left( \pi(\mathbf{x_1} +
\mathbf{x_2}) + Z^n \in
A_{\theta}^{C}(\mathbf{x_1},\mathbf{x_2}) \right)\\
&\stackrel{(a)}\leq& \mbox{Pr}\left( \pi(\mathbf{x_1} +
\mathbf{x_2}) + Z^n \not \in B_{\theta}(\mathbf{x_1} + \mathbf{x_2})
\right) +  \sum_{\mathbf{x'}  \in
\mathcal{C}_{\sqrt{2P}}^{\Delta}\setminus\{ \mathbf{x_1} +
\mathbf{x_2} \}} \mbox{Pr}\left( \pi(\mathbf{x_1} + \mathbf{x_2}) +
Z^n
\in B_{\theta}(\mathbf{x'}) \right)\\
&\stackrel{(b)} = &  \mbox{Pr}\left( \mathbf{t_0} + Z^n \not \in
B_{\theta}(\mathbf{t_0}) \right)  + \sum_{\mathbf{x'} \in
\mathcal{C}_{\sqrt{2P}}^{\Delta}\setminus\{ \mathbf{x_1} +
\mathbf{x_2} \}} p_\theta\left( \mathbf{x_1} + \mathbf{x_2}
,\mathbf{x'}
 \right)\\
&\stackrel{(c)} = &  \mbox{Pr}\left( \mathbf{t_0} + Z^n \not \in
B_{\theta}(\mathbf{t_0}) \right)  + \sum_{g  \in \Lambda_n \setminus
\{0 \}} p_\theta\left(\mathbf{x_1} + \mathbf{x_2}, g +\mathbf{x_1} +
\mathbf{x_2}
 \right)\chi_{T_{\sqrt{2P}}^{\Delta}}(g + \mathbf{x_1} + \mathbf{x_2})\\
\end{eqnarray*}

In the above (a) follows because we use the union bound. In (b), the
first term follows, because due to symmetry, the probability is not
dependent on the particular $\mathbf{x_1} + \mathbf{x_2}$.
 The second term follows as we define $p_\theta(\xv,\mathbf{x'}) $ as
$p_\theta(\xv,\mathbf{x'}) = \mbox{Pr}( \pi(\xv) + Z^n \not \in
B_{\theta}(\mathbf{x'}))$. In (c), we replace $\mathbf{x'}$, by $g +
\mathbf{x_1} + \mathbf{x_2}$ and the characteristic function
$\chi_{T_{\sqrt{2P}}^{\Delta}}$, corresponds to lattice points on
the thin shell at radius $\sqrt{n2P}$.

Hence the average probability of error can be bounded as
\begin{equation}\label{eqn:proberrorblichminkow}
\begin{split}
{P}^{\mathcal{C}_{\oplus}} \leq & \mbox{Pr}\left( \mathbf{t_0} + Z^n
\not \in B_{\theta}(\mathbf{t_0}) \right)  +
\frac{M_{\oplus}'}{M_{\oplus}} + \\&
 \frac{1}{M_\oplus}
\sum_{(\mathbf{x_1},\mathbf{x_2}) \in \mathcal{C}_{\oplus}^{\Delta}}
 \sum_{g  \in \Lambda_n \setminus
\{0 \}} p_\theta\left( \mathbf{x_1} + \mathbf{x_2},g + \mathbf{x_1}
+ \mathbf{x_2} \right)\chi_{T_{\sqrt{2P}}^{\Delta}}(g + \mathbf{x_1}
+ \mathbf{x_2})
\end{split}
\end{equation}

The rest of the proof deals with bounding each of the three terms in
the above equation by an arbitrarily small quantity, to make the
probability of error tend to zero as $n \rightarrow \infty$. Below
we briefly explain the requirements.

\begin{itemize}

\item The first term can be made very small by choosing the angle
$\theta$ appropriately. In effect, we need the noise $Z^n$ to be
contained inside the cone $B_{\theta}(\mathbf{t_0})$ with high
probability as the dimension $n$ becomes large.

 \item For the
second term we need the number of codeword pairs whose sum of
codewords lies outside the thin spherical shell must be shown to be
much lesser than the total number of codeword pairs. In other words
we need to show that the sum of lattice points are concentrated in
the thin spherical shell around the radius $\sqrt{n2P}$.  This is
shown in Lemma~\ref{lem:Hyper Volume Concentration Lemma}.

\item The third term has a summation which is difficult to
evaluate and hence we bound it by an integral and evaluate
the resulting integral.

\item Finally, we require that the number of codewords in each of
the inner sphere must be sufficiently large to achieve rates close
to $\frac{1}{2} \log(\frac12 + SNR)$.

\end{itemize}

  The Blichfeldt's principle(see Theorem~\ref{thm:blichfeldt}) can
  be applied to show concentration of codeword pairs. Lemma~\ref{lem:blichtranslationrate}
  in Appendix~C, is an application of
  the Blichfeldt's principle that guarantees that for any given
  lattice, we can find translations that satisfy $\frac{M_{\oplus}'}{M_{\oplus}} \leq 4
  \frac{V_{\oplus}'}{V_{\oplus}}$. Also it makes sure that we can
  find enough codewords in the hyper spheres of radius $\sqrt{nP}$,
  such that we can achieve a rate of $\frac{1}{2} \log(\frac{1}{2} + SNR)$.

  The Minkowski-Hlawka theorem (see Theorem~\ref{thm:Minkowski-hlawka}
  in Appendix~A) is used to establish
  the existence of at least one lattice such that the
  summation of the third term can be bounded by an integral. This
  theorem along with Lemma~\ref{lem:blichtranslationrate} in Appendix~C,
  are used together in Lemma~8 to obtain bounds on both the
  second and third term. Hence we can effectively rewrite
  (\ref{eqn:proberrorblichminkow}) by using these bounds to get,

\begin{equation}
\begin{split}
{P}^{\mathcal{C}_{\oplus}} & \leq 4
\frac{V_{\oplus}'}{V_{\oplus}^{\Delta}} + \mbox{Pr}\left(
\mathbf{t_0} + Z^n \not \in B_{\theta}(\mathbf{t_0}) \right)  \\ & +
\left[2
 \frac{(n-1)\pi^{\frac{n - 1}{2}}{(n(2P+\delta))^{n/2}}}{d_n n \Gamma (\frac{n +
 1}{2})}   \int_{0}^{\theta} \sin^{n-2}(x) dx\right]
\end{split}
\end{equation}

We can bound the integral, as shown in
\cite[p.~623--624]{shannon59} to get,

\begin{equation}
\begin{split}
{P}^{\mathcal{C}_{\oplus}} & \leq  \mbox{Pr}\left( \mathbf{t_0} +
Z^n \not \in B_{\theta}(\mathbf{t_0}) \right) +  4
\frac{V_{\oplus}'}{V_{\oplus}^{\Delta}}     + \left[2
 \frac{\pi^{\frac{n - 1}{2}}{(n(2P+\delta))^{n/2}}}{d_n n \Gamma (\frac{n +
 1}{2})}  \frac{ \sin^{n-1}(\theta)}{\cos \theta}\right]
\end{split}
\end{equation}

Now we next need to choose the appropriate values for $\theta$ and
$d_n$ to make the probabilities go to 0. For the second term,
consider $\sin \angle (t_0 + Z,t_0)$ which is given by
\begin{equation*}
\sin \angle (t_0 + Z,t_0) =
\sqrt{\frac{\|Z^\bot\|^2}{(\sqrt{n(2P-\delta)} + Z^{\|})^2 +
\|Z^\bot \|^2}} \stackrel{n\rightarrow \infty}\longrightarrow
\sqrt{\frac{\sigma^2}{2P - \delta + \sigma^2}}
\end{equation*}
Hence we choose $\sin \theta = \sqrt{\frac{\sigma^2}{2P - \delta +
\sigma^2}}.$ For the third term a good choice of $d_n$ is
\begin{equation*}
d_n = V_n(\sqrt{n(2P+2\delta)}) \sin^n \theta =  \frac{ \pi^{n/2}
{(n(2P + 2 \delta))^{n/2}} (\sin^n\theta)}{\Gamma(n/2 + 1)}
\end{equation*}
This choice helps us to make the third term tend to 0 for large $n$.
The third term then can be rewritten as given below. We use the
results in  \cite[p. 277]{urbanke98}, to bound the Gamma functions
to get,
\begin{equation*}
\begin{split}
\left[2
 \frac{\pi^{\frac{n - 1}{2}}{(n(2P+\delta))^{n/2}}}{d_n n \Gamma (\frac{n +
 1}{2})}  \frac{ \sin^{n-1}(\theta)}{\cos \theta}\right]
  = \frac{2}{\sqrt{\pi} \sin \theta \cos \theta}
  \left(\frac{2P+\delta}{2P + 2\delta}\right)^{n/2}
  \frac{\Gamma(\frac{n}{2} + 1)}{n \Gamma(\frac{n+1}{2})}\\
  <\frac{2}{\sqrt{\pi} \sin \theta \cos \theta}
  \left(\frac{2P+\delta}{2P + 2\delta}\right)^{n/2} \frac{1}{2}
  \frac{\left[1 + \Gamma(\frac{n+1}{2})\right]}{ \Gamma(\frac{n+1}{2})}
 \end{split}
\end{equation*}
This decays to 0 exponentially as $n \rightarrow \infty$.

Now, the achievable rate
can be obtained from the number of lattice points in the sphere of
radius $\sqrt{nP}$. This value $M(\Lambda_n,\mathbf{s_1^*} )$,
$M(\Lambda_n,\mathbf{s_2^*} )$, from the lemma,  can be seen to be
greater than $\frac{V}{8 d_n}$. Hence the rate $R$ is given by,

\begin{equation*}
\begin{split}
 R & \geq \frac{1}{n} \log M(\Delta_n,\mathbf{s_1^*}) \\
 & \geq \frac{1}{n} \log \frac{V}{8 d_n} \\
& \geq \frac{1}{n} \log \left[\frac{1}{8} \left(\frac{P}{(2P + 2
\delta) \sin^2 \theta}\right)^{n/2}\right]\\
& \geq \frac{1}{2} \log \left( \frac{P}{2P + 2 \delta} + \left(
\frac{2P - \delta}{2P + 2\delta}\right) \frac{P}{\sigma^2}  \right)
-
\frac{\log 8}{n}\\
& \geq \frac{1}{2} \log \left( \frac{1}{2} +\frac{P}{\sigma^2}
\right) - \delta'
\end{split}
 \end{equation*}

Choosing $\delta'$ arbitrarily small, a rate of $\frac{1}{2} \log
\left( \frac{1}{2} +\frac{P}{\sigma^2} \right)$ can be achieved.

 \subsection{Relationship with ML decoder} There are some
 conditions under which the the minimal angle decoder
 will perform like the ML decoder.
\begin{itemize}
\item[(1)]
  It can be
easily seen that, for Gaussian noise, the ML decoder is equivalent to
minimum distance decoder.
\item[(2)]
If all the codewords are concentrated (with high probability) in the
thin shell then, we do not lose much by neglecting the codewords
outside the thin shell.
\item[(3)] Suppose the width of the thin
shell is very small and almost all the codewords have
approximately the same distance from the origin. Then, calculating
minimum distance from received vector to the codewords is equivalent
to calculating the minimum angle.
\item[(4)]
Suppose almost all the lattice points on the thin shell are codewords.
Then, decoding to any lattice point in the thin shell does not sacrifice
performance.
\end{itemize}

We are dealing with Gaussian noise so the first condition is easily
satisfied. In the course of our proof, we will observe by applying
the Blichfeldt's principle that there exists a concentration of
codewords at the thin shell. Hence condition (2) also holds.
Moreover, we will also let the width of the thin shell become
arbitrarily small, hence what we are doing is very close to ML
decoding. The 4th condition appears not hold to at low SNR, however.
In this case, all lattice points in the thin shell may not be
codewords. Hence, we think this may lead to the sub-optimality at low
SNR. The theorem shows that, for SNR $< 1/2$, we get zero
rate. We know that, for random Gaussian codebooks, joint decoding
of both codewords is possible though.  Therefore, we think that
the ML decoder may give a better performance at low SNR.

\section{Joint Decoding Based Scheme} \label{sec:jointdecoding}

While the aforementioned scheme is nearly optimal at high SNR, the
performance of this scheme at low SNR is very poor. In fact, for
$P/\sigma^2 < 1/2$, the scheme does not even provide a non-zero
rate. In this regime, we can use any coding scheme which is optimal
for the multiple access channel and perform joint decoding at the
relay such that $\uv_A$ and $\uv_B$ can be decoded. Then, the relay
can encode $\uv_A \oplus \uv_B$ and transmit to the nodes. Any
coding scheme that is optimal for the MAC
 channel can be used in the MAC phase.
  A simple scheme is time sharing (although it is not the only one) where nodes $A$ and $B$
transmit with powers $2P$ for a duration of $n/2$ channel uses each
but they do not interfere with one another. In this case, a rate of
\begin{equation}
R_{JD} = \frac{1}{4} \log \left(1 + \frac{2P}{\sigma^2} \right)
\end{equation}
can be obtained. It can be seen that this is optimal at
asymptotically low SNR, since $\log(1+{\sf snr}) \approx {\sf snr}$
for ${\sf snr} \rightarrow 0$. \vspace{1cm}

 \begin{figure}[htbp]
      \centering
   \includegraphics[width=4.0in]{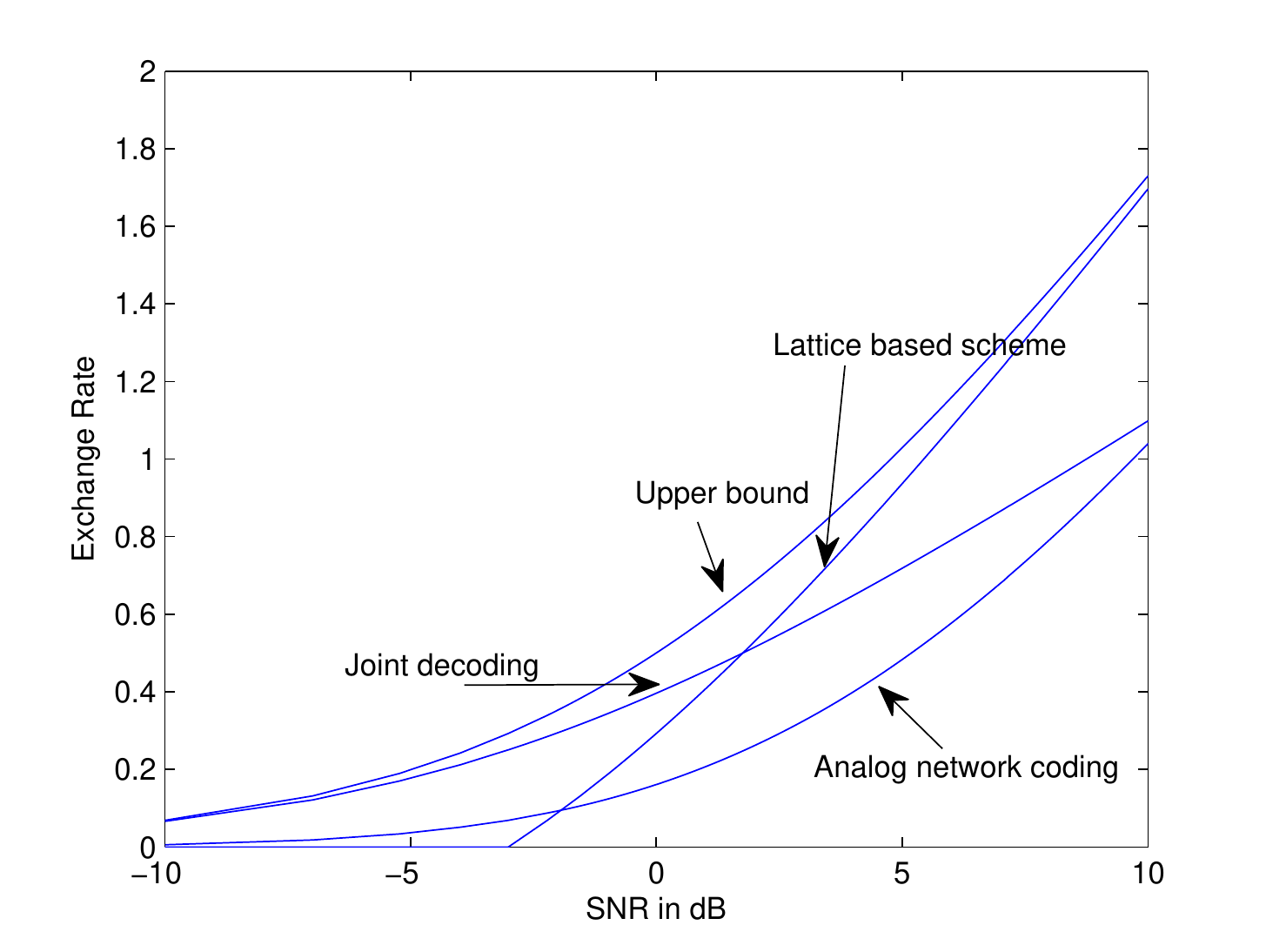}
      \caption{Achievable Exchange Rates for the Proposed Schemes}
      \label{fig:performance}
   \end{figure}

The performance of these schemes is shown in
Fig.~\ref{fig:performance} where the upper bound and the achievable
rates with these proposed schemes are shown. The achievable rate
with the analog network coding scheme is also shown. It can be seen
that our schemes outperform analog network coding. However, it must
be noted here that the scheme proposed here requires perfect synchronization of the
phases from the transmissions, whereas the analog network coding scheme
does not.

Clearly, we can time share between the lattice based scheme and the joint
decoding based scheme in order to obtain rates of the form
\begin{equation}
R_{ex} = \beta R_{ex,JD} + (1-\beta) R_{ex,Lattice}.
\end{equation}

It can be shown that between the SNRs of -0.659~dB and 3.46~dB, time sharing between the
two schemes results in better rates than the individual schemes.
In both the lattice based scheme and the joint decoding based
scheme, if the restriction to use $n$ channel uses during the MAC
phase and $n$ during the broadcast phase is removed, i.e., only the
total number of uses is constrained to be $2n$,  better schemes can
be easily designed. Similarly, different power sharing may also lead
to better schemes.

\section{Extension to multiple hops}

\subsection{Description}

We can extend the above results to multiple hops. We can again show
that rate of $\frac{1}{2} \log\left(\frac{1}{2} +
\frac{P}{\sigma^2}\right)$ is achievable using structured coding
even in the multiple hop scenario. It should be noted that the
advantage of this scheme over the amplify and forward
 scheme \cite{sachin07} becomes more pronounced in the multi-hop case, since at each stage for the amplify and forward scheme,
 the channel noise is amplified and
hence the amplify and forward scheme will suffer a huge rate loss as
the number of hops increase. The problem model is shown in
Fig.~\ref{fig:multihop}. The relay nodes and the nodes $A$ and $B$
can transmit only to the two nearest nodes. During a single
transmission slot ($n$ uses of the channel), a node can either
listen or transmit. That is, it can not do both simultaneously. We
explain our structured coding scheme using a simple example of a
3-relay network. The different transmissions are shown in the table
given below.

Here node $A$ and node $B$ have data that need to be exchanged
between each other. Each node has a stream of packets. Node $A$ has
packets named $\uv_{A,1},\uv_{A,2},\ldots$ and node $B$ has packets
named $\uv_{B,1},\uv_{B,2},\ldots$. In the first transmission slot
the nodes $A$ and $B$ transmit. Nodes $A$, $B$ transmit vectors
$\xv_{A,1}$ and $\xv_{B,1}$, respectively using our proposed lattice
coding scheme. At the beginning of transmission, the node $R_2$ has
no data to transmit in the first transmission slot, and hence it
remains silent. The node $R_1$ and $R_3$ decode to $\xv_{1,1} \mod
\Lambda$ and $\xv_{2,1} \mod \Lambda$, respectively. During the
second transmission slot the nodes $R_1$ and $R_3$ transmit, while
the other nodes remain silent.  So, in each stage the nodes transmit
and the listening nodes decode to a lattice point in the fine
lattice in the Voronoi region of the coarse lattice.
 In every second transmission slot
a new packet is transmitted to the relay nodes by the nodes $A$ and
$B$ as can be seen from Table~\ref{table:multihop}. During slots
$2,4,6$ nodes $A,B$ transmits new packets into the relay channel.
From this example, we can see that at the 4th slot the node $A$ and
$B$ decode $\xv_{B,1}$ and $\xv_{1,2}$ respectively. This is because
the node $A$ receives $(\xv_{1,1}+\xv_{1,2}+\xv_{2,1}) \mod \Lambda$
during the 4th transmission and, hence, since $\xv_{1,1}, \xv_{1,2}$
are already known at the node $A$, the node $A$ can decode to
$\xv_{2,1}$ using modulo operation. The same argument holds for node
$B$. From every two transmissions from this stage a new packet can
be decoded at each node. This shows that for sufficiently large
number packets we can achieve the rate of $\frac12 \log(\frac12 +
\frac{P}{\sigma^2})$. A similar encoding scheme can be used for
$L=2$ nodes also, in the first slot, node A and $R_2$ transmit,
while the others listen. In the next slot $R_1$ and node $B$
transmit while the others listen and decode. Again the same rate of
$\frac12 \log(\frac12 + \frac{P}{\sigma^2})$ is achievable.

\subsection{Achievable rate}

\begin{thm}
For the multi hop scenario defined with $L$ hops, an exchange rate
of $\frac{1}{2} \log (\frac12 + \frac{P}{\sigma^2})$ is achievable
with nested lattice encoding and decoding.
\end{thm}

\begin{figure}[thpb]
      \centering
     \includegraphics[width=5.0in]{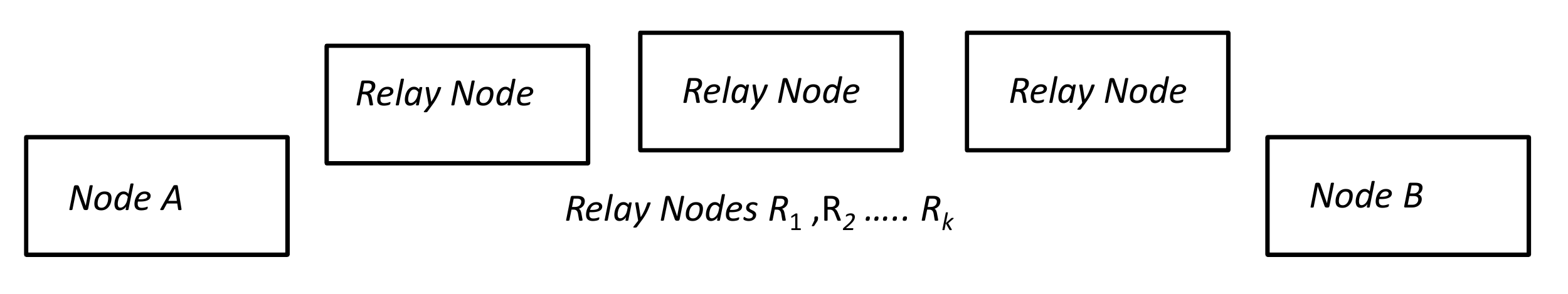}
      \caption{Problem model for multi hop}
      \label{fig:multihop}
   \end{figure}

\begin{table*}
\begin{center}
\caption{Table Showing Sequence of Packets in the Multihop Case}
\begin{tabular}[t]{|c|p{2cm}|p{2.7cm}|p{2.7cm}|p{2.7cm}|p{2.5cm}|}

  \hline
   Slot &Node $A$ & Node $R_1$ $(\mod \Lambda)$ & Node $R_2$ $(\mod \Lambda)$ & Node $R_3$ $(\mod \Lambda)$ &
  Node $B$
  \\ \hline
  1 & Transmits $\xv_{1,1}$ & $\xv_{1,1}$ & Transmits & $\xv_{2,1}$ & Transmits $\xv_{2,1}$ \\
  2 & Remains Silent & Transmits & $\xv_{1,1} + \xv_{2,1}$ & Transmits & Remains Silent \\
  3 & Transmits $\xv_{1,2}$ & $\xv_{1,2}+\xv_{1,1}+\xv_{2,1}$ & Transmits & $\xv_{1,1}+\xv_{2,1}+\xv_{2,2}$ & Transmits $\xv_{2,2}$ \\
  4 & Decodes $\xv_{2,1}$ & Transmits & $2(\xv_{1,1} + \xv_{2,1}) + \xv_{1,2}+\xv_{2,2}$ & Transmits & Decodes $\xv_{1,1}$ \\
  5 & Transmits $\xv_{1,3}$ & $2(\xv_{1,1} + \xv_{2,1})$  $+ \xv_{1,2}+\xv_{2,2} + \xv_{1,3}$ & Transmits & $2(\xv_{1,1}+ \xv_{2,1})$  $+
  \xv_{1,2}+\xv_{2,2}+ \xv_{2,3}$ & Transmits $\xv_{2,3}$ \\
  6 & Decodes $\xv_{2,2}$ & Transmits  & $4(\xv_{1,1} +\xv_{2,1}) + 2(\xv_{1,2}+\xv_{2,2}) +\xv_{1,3} + \xv_{2,3}$ & Transmits & Decodes
  $\xv_{1,2}$ \\
  \hline
\end{tabular}
\label{table:multihop}
\end{center}
\end{table*}

\begin{proof}
      We can easily prove that the theorem holds for a general
      case $L$ relay nodes in between. In our coding
      scheme  in every two slots a new packet is sent out
      from the nodes $A$ and $B$.  After an initial $2L$ transmission slot
      delay, in every two slots the relay
      nodes  receives  a new packet from the other nodes. Here, we
      mean that every two slots the relay node decodes to a lattice
      point which is a linear function of a new packet. Hence, at the
      decoding stage at the nodes $B$ and $A$, we can decode after every
      two slots since only one variable is unknown, since only one new
      packet (or function of new packet) moves from one node to the
      other. Hence, we can still achieve the rate of $\frac12 \log (\frac12 +
      \frac{P}{\sigma^2})$. Moreover, the functions in each stage are
      bounded  for a finite $L$ and, hence, we can always perform the decoding at the
      receiver nodes.

\end{proof}

\section{Conclusion}

We considered joint physical layer and network layer coding for the
bi-directional relay problem where two nodes wish to exchange
information through AWGN channels. Under the restrictive model of
the MAC and broadcast phase using $n$ channel uses separately, we
showed upper bounds on the exchange capacity and constructive
schemes based on lattices that is nearly optimal at high SNR. At
low SNR joint decoding based schemes (optimal coding schemes for
the MAC channel) are nearly optimal. These schemes outperform the
recently proposed analog network coding. Interestingly, our result
shows that structured codes such as lattice codes outperform random
codes for such networking problems. We also showed that minimal
angle decoding also leads to similar results. We also showed
extensions of this scheme to a network with many relay nodes, where
the advantages of the proposed scheme over simple amplify and
forward will be higher than in single relay case.

\useRomanappendicesfalse
\appendices

\section{Blichfeldt's Principle and Minkowski-Hlawka Theorem}

\begin{thm}[Blichfeldt's Principle \cite{Minkowski-Hlawka}]\label{thm:blichfeldt}
Let $f$ be a Riemann integrable function with bounded support. If
$\Lambda_n$ is a lattice with fundamental region $P_n$ then
\begin{equation*}
\int_{\mathds{R}^n} f(\sv) dV(\sv) = \int_{P_n} \left( \sum_{h \in
\Lambda_n} f(h + \sv) \right) dV(\sv).
\end{equation*}

\end{thm}

 Let us define the following function $ V_{\oplus}$ as
 follows,
 \begin{equation*}
 V_{\oplus} = \int \int \chi_{T}(\uv) \chi_{T}(\vv)
  dV(\uv)dV(\vv).
 \end{equation*} Here $ V_{\oplus} = (V_n(\sqrt{nP}))^2$, represents the square of
 the volume of an $n$-dimensional sphere of radius $\sqrt{nP}$. $dV$
 represents the $n$-dimensional volume element in rectangular co-ordinates.

We next establish the following corollary

\begin{cor}\label{cor:blichapp}
\begin{equation*}
 V_{\oplus} = \int_{\mathbf{s_1}} \int_{\mathbf{s_2}} M_{\oplus}(\Lambda_n, \mathbf{s_1},\mathbf{s_2})
  dV(\mathbf{s_1})dV(\mathbf{s_2}).
 \end{equation*}
\end{cor}

\begin{proof} Let us define a function
\begin{equation} f(\uv,\vv) = \chi_{T}(\uv) \chi_{T}(\vv).
\end{equation}
For a fixed $\uv$, $f(\uv,\vv)$ can be seen as a function with
bounded support and also can be seen to be integrable. Hence we can
apply the Blichfeldt's principle to get
\begin{equation*}
h(\uv) = \int f(\uv,\vv) dV(\vv) = \int_{P_n} \left( \sum_{h_2 \in
\Lambda_n} f(\uv, h_2 + \mathbf{s_2}) \right) dV(\mathbf{s_2}).
\end{equation*} Now $h(\uv)$ can again be seen as  a Riemann
integrable function with bounded support, and hence the Blichfeldt's
principle could be applied again to get the following,
\begin{eqnarray*}
 V_{\oplus} &=& \int h(\uv) dV(\uv) \\
  &=& \int_{P_n} \sum_{h_1 \in \Lambda_n} \left( \int_{P_n} \left( \sum_{h_2 \in
\Lambda_n} f(h_1 + \mathbf{s_1}, h_2 + \mathbf{s_2}) \right)
dV(\mathbf{s_2})\right) dV(\mathbf{s_1})\\
 &\stackrel{(a)}=& \int_{P_n}   \int_{P_n} \left( \sum_{h_1 \in \Lambda_n} \sum_{h_2 \in
\Lambda_n} f(h_1 + \mathbf{s_1}, h_2 + \mathbf{s_2})
\right) dV(\mathbf{s_1})dV(\mathbf{s_2})\\
&\stackrel{(b)}=& \int_{P_n}   \int_{P_n} M_{\oplus}(\Lambda_n,
\mathbf{s_1},\mathbf{s_2})
 dV(\mathbf{s_1})dV(\mathbf{s_2})\\
 \end{eqnarray*}

 Above (a) follows since we have a finite number of non-zero terms,
 and hence the integral and the summation can be interchanged. Also
 in (b), $M_{\oplus}(\Lambda_n,
\mathbf{s_1},\mathbf{s_2})$ is the number of pairs of lattice points
for the translations $\mathbf{s_1}$ and $\mathbf{s_2}$.

\end{proof}

In short the Corollary 8 relates the square of the volume of an
$n$-dimensional sphere and number of pairs of lattice points for
different translations.

\begin{thm}[Minkowski-Hlawka]\label{thm:Minkowski-hlawka} Let $f$ be a nonnegative Riemann
integrable function with bounded support. Then for every $d \in
\mathds{R}^+$ and $n \geq 2$, there exists a lattice $\Lambda_n$
with determinant $det(\Lambda_n) = d$ such that
\begin{equation*}
d \sum_{g \in \Lambda_n\setminus \{0\}} f(g) \leq \int f dV(\xv)
\end{equation*}

\end{thm}

 The Minkowski-Hlawka theorem gives us a way to connect a series of
 discrete sums with a continuous integral. This will find
 applications in our probability of error calculations.

\section{Hyper Volume Concentration Lemma}

\begin{lem}\label{lem:Hyper Volume Concentration Lemma}
 Let $V_{\oplus}'$ be defined as
 \begin{equation*}
 V_{\oplus}^\prime = \int \int \chi_T(\uv) \chi_T(\vv) \chi_{T^\prime_{\sqrt{2P}}}(\uv + \vv) dV(\uv)
 dV(\vv)
  \end{equation*} then  we can choose
  $n$ sufficiently large such that,
   $ \frac{V_{\oplus}'}{V_{\oplus}} <
   \delta,
   $ for every given positive $\delta$.

\end{lem}

\begin{proof}

First we perform a change of variables in the integral, by
substituting $\xv = \uv + \vv$. This gives,
\begin{equation*}
 V_{\oplus}^\prime = \int \int \chi_T(\uv) \chi_T(\xv -\uv) \chi_{T^\prime_{\sqrt{2P}}}(\xv) dV(\uv)
 dV(\xv)
  \end{equation*}

\begin{figure}[htbp]
\begin{center}
\includegraphics[scale=0.4,angle = 0]{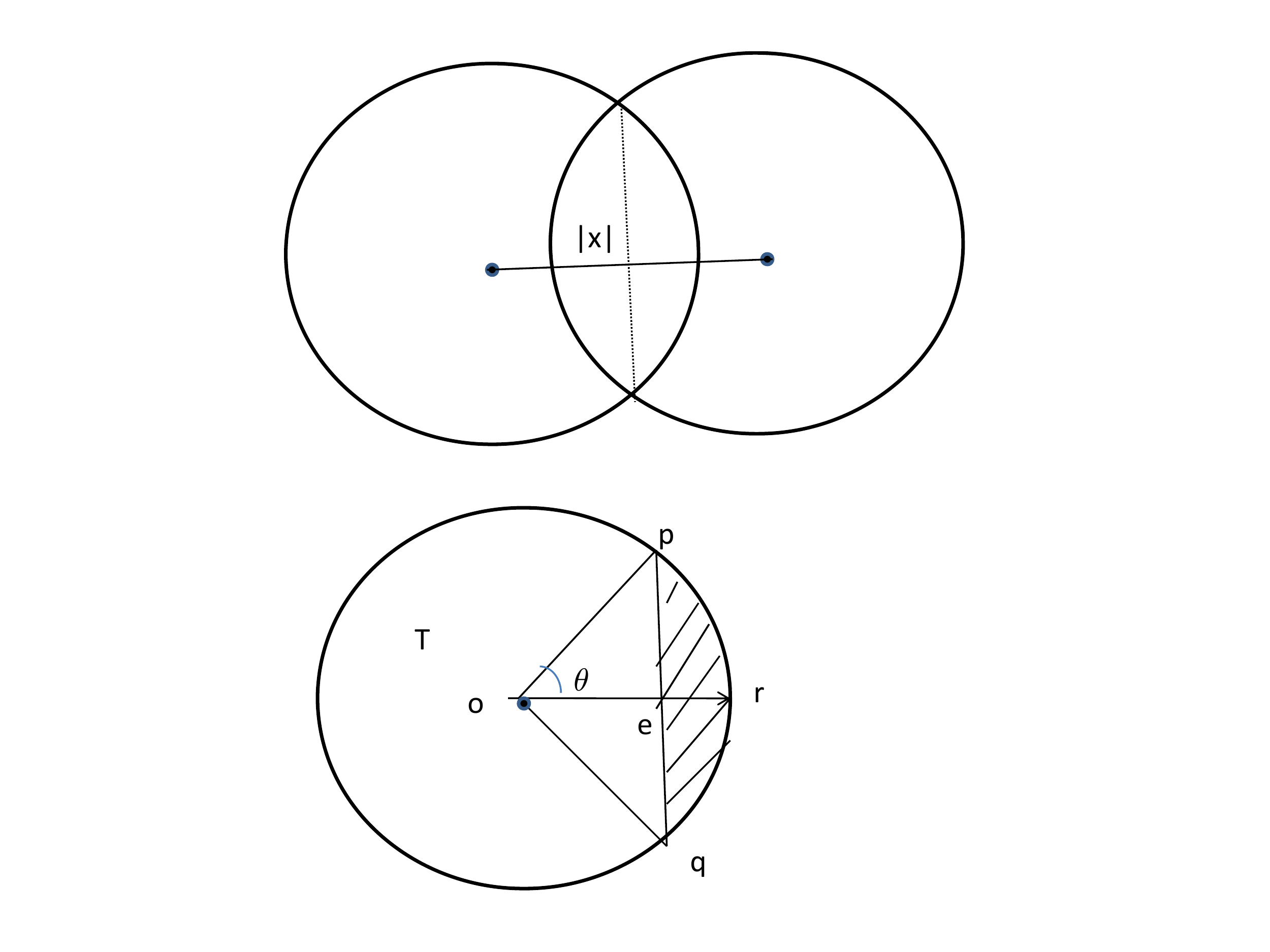}
\end{center}
\caption{Geometrical interpretation of the integral}
\label{fig:conv}\end{figure}

Let us consider first the inner integral, for a fixed $\xv$,given
by,
\begin{equation*} \int \chi_T(\uv) \chi_T(\xv -\uv)  dV(\uv).
\end{equation*} This geometrically represents the hyper volume of
intersection of two hyper-spheres, whose centers are at a distance
$\| \xv \|$, from each other. This is pictorially shown in
Fig.~\ref{fig:conv}. The calculation of hyper volume of
intersection, reduces to obtaining the hyper volume of the conical
section and a cone. This is shown pictorially in the second diagram
in Fig.~\ref{fig:conv}. Here $opq$ represents the hyper cone and
$oprq$ represents the conical section. Here we denote by
$V_{cs}(|\xv|)$ to represent the volume of the conical section and
$V_{co}(|\xv|)$ as the volume of the cone. The integral can hence be
be evaluated as
\begin{equation*}\int_T \chi_T(\vv) \chi_T( \mathbf{x} - \mathbf{v})
 d\mathbf{v} = 2 (V_{cs}(|\xv|) - V_{co}(|\xv|)).\end{equation*}

To simplify calculations we can bound the integral as,
\begin{equation*}\int_T \chi_T(\vv) \chi_T( \mathbf{x} - \mathbf{v})
 dV(\mathbf{v}) \leq 2 V_{cs}(|\xv|) .\end{equation*}

From the figure we can see $oe = \|x\| /2 $. Hence the half-angle
$\theta$ can be calculated as $\cos \theta = \frac{oe}{\sqrt{nP}} =
\frac{\|x\|}{2 \sqrt{nP}}.$ Hence with the defined $\theta$, we can
calculate the hyper volume as,

\begin{equation}
V_{cs}(|\xv|) = 2 \left(\int_{\psi=0}^{\theta} \sin^{n-2} \psi
d\psi\right) \frac{\pi^{\frac{n-1}{2}}(nP)^{\frac{n}{2}}}{n
\Gamma(\frac{n-1}{2})}
\end{equation}

Hence,

\begin{equation}
\frac{V_{\oplus}'}{V_\oplus} \leq 2  \int
\frac{V_{cs}(\|\xv\|)}{V_\oplus} \chi_{T^\prime_{\sqrt{2P}}}(\xv)
dV(\xv)
\end{equation}

But $V_\oplus$ is given by $V^2$, where $V$ is the hyper volume of
an n-dimensional hyper sphere of radius $\sqrt{nP}$ denoted by
$V_n(\sqrt{nP}).$

The value of $V_{cs}(\|x\|)$  depends only on the distance of $\xv$
from the origin. To make evaluation of the integral easier, we change the
volume element to circular co-ordinates and integrate. Thus the
integral now becomes,

\begin{equation}
\frac{V_{\oplus}'}{V_\oplus} \leq   2 \left( \frac{n
\pi^{\frac{n}{2}}}{\Gamma(\frac{n}{2}+1)} \right) \int_{r \in
l_{\oplus}'}  \left. \frac{V_{cs}(r)}{V_n({\sqrt{nP}})^2} \right.
r^{n-1} dr,
\end{equation}
where $l_{\oplus}'$ is defined as the union of closed
intervals,(will be given later). Now substituting $V_{cs}(r)$ from
above gives,

\begin{equation}
\frac{V_{\oplus}'}{V_\oplus} \leq 4 \frac{n
\pi^{\frac{n}{2}}}{\Gamma(\frac{n}{2}+1)}
\frac{\pi^{\frac{n-1}{2}}(nP)^{\frac{n}{2}}}{n
\Gamma(\frac{n-1}{2})} \int_{r \in l_{\oplus}'}
\int_{\psi=0}^{\cos^{-1}(\frac{r}{2 \sqrt{nP}})}
  \frac{r^{n-1} \sin^{n-2} \psi}{V_n({\sqrt{nP}})^2} d\psi dr
\end{equation}

Now let us choose $\cos \theta = \frac{r}{2 \sqrt{nP}}$. Then change
of variables gives

\begin{equation}
\frac{V_{\oplus}'}{V_\oplus} \leq 4 \frac{n
\pi^{\frac{n}{2}}}{\Gamma(\frac{n}{2}+1)}
\frac{\pi^{\frac{n-1}{2}}(nP)^{\frac{n}{2}}}{n
\Gamma(\frac{n-1}{2})} \int_{\theta \in \theta_{\oplus}'} (nP)^{n/2}
\cos^{n-1} \theta  \sin \theta \int_{\psi=0}^{\theta}
  \frac{2^{n} \sin^{n-2} \psi}{V_n({\sqrt{nP}})^2} d\psi d \theta
\end{equation}

Substituting for $V_n(\sqrt{nP})$, we get,
\begin{equation}
\frac{V_{\oplus}'}{V_\oplus} \leq \frac{4}{\sqrt{\pi}}
\frac{\Gamma(\frac{n}{2}+1)}{\Gamma(\frac{n-1}{2})} \int_{\theta \in
\theta_{\oplus}'} 2^{n} \cos^{n-1} \theta \sin \theta
\int_{\psi=0}^{\theta}
   \sin^{n-2} \psi d\psi d \theta
\end{equation}

\begin{equation}
\begin{split}
\frac{V_{\oplus}'}{V_\oplus} \leq \frac{4}{\sqrt{\pi}}
\frac{\Gamma(\frac{n}{2}+1)}{\Gamma(\frac{n-1}{2})} \int_{\theta \in
\theta_{\oplus}' \cap [0, \frac{\pi}{3} + \eta ]} 2^{n} \cos^{n-1}
\theta \sin \theta \int_{\psi=0}^{\theta}
   \sin^{n-2} \psi d\psi d \theta \\
    + \frac{4}{\sqrt{\pi}}
\frac{\Gamma(\frac{n}{2}+1)}{\Gamma(\frac{n-1}{2})} \int_{\theta \in
\theta_{\oplus}' \cap [ \frac{\pi}{3} + \eta , \frac{\pi}{2}]} 2^{n}
\cos^{n-1} \theta \sin \theta \int_{\psi=0}^{\theta}
   \sin^{n-2} \psi d\psi d \theta
\end{split}
\end{equation}

Now we will use the bound by Shannon for the first term. We can
apply the bound for $\theta < \pi /2$. Hence we  split the integral
into two terms to apply the bound. For the second term we will bound
$\sin \psi$ by $1$. This gives,

\begin{equation}
\begin{split}
\frac{V_{\oplus}'}{V_\oplus} \leq \frac{4}{\sqrt{\pi}}
\frac{\Gamma(\frac{n}{2}+1)}{\Gamma(\frac{n-1}{2})} \int_{\theta \in
\theta_{\oplus}' \cap [0, \frac{\pi}{3} + \eta ]} 2^{n} \cos^{n-1}
\theta \sin \theta \frac{\sin^{n-1} \theta}{(n-1) \cos \theta}
 d \theta \\
    + \frac{4}{\sqrt{\pi}}
\frac{\Gamma(\frac{n}{2}+1)}{\Gamma(\frac{n-1}{2})} \int_{\theta \in
\theta_{\oplus}' \cap [ \frac{\pi}{3} + \eta , \frac{\pi}{2}]} 2^{n}
\cos^{n-1} \theta \sin \theta \int_{\psi=0}^{\theta}
    d\psi d \theta
\end{split}
\end{equation}

Simplifying things further we get,

\begin{equation}
\begin{split}
\frac{V_{\oplus}'}{V_\oplus} \leq \frac{4}{\sqrt{\pi}}
\frac{\Gamma(\frac{n}{2}+1)}{\Gamma(\frac{n+1}{2})} \int_{\theta \in
\theta_{\oplus}' \cap [0, \frac{\pi}{3} + \eta ]} 2^{n-1} \sin^{n-1}
\theta \cos^{n-1} \theta \tan \theta
 d \theta \\
    + \frac{4}{\sqrt{\pi}}
\frac{(n-1)\Gamma(\frac{n}{2}+1)}{\Gamma(\frac{n+1}{2})}
\int_{\theta \in \theta_{\oplus}' \cap [ \frac{\pi}{3} + \eta ,
\frac{\pi}{2}]} 2^{n-1} \cos^{n-1} \theta \sin \theta \frac{\pi}{2}
d \theta
\end{split}
\end{equation}

Next we bound $\tan \theta$ in the first term by  $\tan (\pi/3 +
\eta)$ and $\sin \theta$ in the second term by $1$. We also bound
the factorial using the bound given in Urbanke. This yields,

\begin{equation}
\begin{split}
\frac{V_{\oplus}'}{V_\oplus}\leq \frac{4}{\sqrt{\pi}}
\frac{\frac{n}{2}(1 + \Gamma(\frac{n+1}{2})) \tan (\frac{\pi}{3} +
\eta)}{\Gamma(\frac{n+1}{2})} \int_{\theta \in \theta_{\oplus}' \cap
[0, \frac{\pi}{3} + \eta ]} (\sin 2\theta)^{n-1}
 d \theta \\
    + \frac{4}{\sqrt{\pi}}
\frac{(n-1)\frac{n}{2}(1 +
\Gamma(\frac{n+1}{2}))}{\Gamma(\frac{n+1}{2})} \int_{\theta \in
\theta_{\oplus}' \cap [ \frac{\pi}{3} + \eta ,
\frac{\pi}{2}]}\frac{\pi}{2} (2 \cos \theta)^{n-1}  d \theta
\end{split}
\end{equation}

Again we can see that since $\theta$ does not take the value $\pi/4$
we can bound the first term appropriately. In the second term the
maximum value of $\cos \theta$ can be used to  bounded it
appropriately. This is given below as follows.

\begin{equation}
\begin{split}
\frac{V_{\oplus}'}{V_\oplus} \leq \frac{4}{\sqrt{\pi}}
\frac{\frac{n}{2}(1 + \Gamma(\frac{n+1}{2})) \tan (\frac{\pi}{3} +
\eta)}{\Gamma(\frac{n+1}{2})}  (\sin (\frac{\pi}{2} - 2
\epsilon))^{n-1} \frac{\pi}{2}
\\
    + \frac{4}{\sqrt{\pi}}
\frac{(n-1)\frac{n}{2}(1 +
\Gamma(\frac{n+1}{2}))}{\Gamma(\frac{n+1}{2})} \frac{\pi}{2} [2 \cos
(\frac{\pi}{3} + \eta)]^{n-1} \frac{\pi}{2}
\end{split}
\end{equation}

From the above we can easily see, that both terms tend to $0$ as $n
\rightarrow \infty$. From this the lemma follows.

\end{proof}

\section{Application of Blichfeldt's Principle to Show Existence of Good Translations}

\begin{lem}\label{lem:blichtranslationrate} Let $\Lambda_n$ be a lattice having a fundamental
region $P_n$ and let $\det(\Lambda_n)$ be it's fundamental volume
and define
\begin{equation*}
\begin{split}P_n^* = \left\{ (\mathbf{s_1},\mathbf{s_2}) \in P_n \times P_n :
M_1(\Lambda_n,\mathbf{s_1}) \geq \frac{V_n(\sqrt{nP})}{ 8
\det(\Lambda_n)}; \right.\\ \left. M_2(\Lambda_n,\mathbf{s_2}) \geq
\frac{V_n(\sqrt{nP})}{8 \det(\Lambda_n)};
\frac{M_\oplus'(\Lambda_n,\mathbf{s_1},\mathbf{s_2})}{M_\oplus(\Lambda_n,\mathbf{s_1},\mathbf{s_2})}\leq
4 \frac{V_\oplus'}{V_\oplus} \right\}
\end{split}
\end{equation*}
Then \begin{equation*}V_{\oplus} \leq 2 \int_{P_n^*}
M_{\oplus}(\Lambda_n, \mathbf{s_1},\mathbf{s_2})
 dV(\mathbf{s_1},\mathbf{s_2}),\end{equation*} where $\delta_n > 0$
and can be made arbitrarily small for sufficiently large $n$
\end{lem}

\begin{proof}
Let us define the following sets,
\begin{equation*}
F_n = \left\{ \mathbf{s_1} \in P_n: M_1(\Lambda_n,\mathbf{s_1}) \geq
\frac{V_n(\sqrt{nP})}{8 \det(\Lambda_n)} \right\}
\end{equation*}

\begin{equation*}
G_n = \left\{ \mathbf{s_2} \in P_n: M_2(\Lambda_n,\mathbf{s_2}) \geq
\frac{V_n(\sqrt{nP})}{8 \det(\Lambda_n)} \right\}
\end{equation*}

\begin{equation*}
O_n = \left\{ (\mathbf{s_1},\mathbf{s_2}) \in P_n \times P_n:
 \frac{M_{\oplus}'(\Lambda_n,
\mathbf{s_1},\mathbf{s_2})}{M_{\oplus}(\Lambda_n,
\mathbf{s_1},\mathbf{s_2})}  \leq 4 \frac{ V_\oplus'}{V_\oplus}
\right\}
\end{equation*}

Therefore $P_n^* = (F_n \times G_n) \bigcap O_n .$ Define the
complements here. Hence
\begin{equation}
P_n \times P_n = P_n^* \bigcup \left\{(F_n \times G_n) \cap O_n^C
\right\} \bigcup \left\{(F_n \times G_n)^C \right\}
\end{equation}
Hence we have
\begin{equation}
V_{\oplus} = \int_{P_n \times P_n} M_{\oplus}(\Lambda_n,
\mathbf{s_1},\mathbf{s_2})
 dV(\mathbf{s_1})dV(\mathbf{s_2})
 \end{equation}
 \begin{equation*}
 \begin{split}
  V_{\oplus}  \leq &
 \int_{P_n^*}
M_{\oplus}(\Lambda_n, \mathbf{s_1},\mathbf{s_2})
 dV(\mathbf{s_1})dV(\mathbf{s_2})
 + \int_{O_n^C}
M_{\oplus}(\Lambda_n, \mathbf{s_1},\mathbf{s_2})
 dV(\mathbf{s_1})dV(\mathbf{s_2}) \\
 & + \int_{(F_n\times G_n)^C}
M_{\oplus}(\Lambda_n, \mathbf{s_1},\mathbf{s_2})
 dV(\mathbf{s_1})dV(\mathbf{s_2})
 \end{split}
 \end{equation*}

We can replace the second integral, using the condition that  the
translations are not in the set $O_n$, to get a bound on
$M_{\oplus}$ as shown below,

 \begin{equation*}
  \begin{split}
V_{\oplus}  \leq &
 \int_{P_n^*}
M_{\oplus}(\Lambda_n, \mathbf{s_1},\mathbf{s_2})
 dV(\mathbf{s_1})dV(\mathbf{s_2})
  + \frac{V_{\oplus}}{4 V_{\oplus}'}\int_{O_n^C}
M_{\oplus}'(\Lambda_n, \mathbf{s_1},\mathbf{s_2})
 dV(\mathbf{s_1})dV(\mathbf{s_2})\\
 & + \int_{(F_n\times G_n)^C}
\left( \sum_{h_1 \in \Lambda_n} \sum_{h_2 \in \Lambda_n} f(h_1 +
\mathbf{s_1}, h_2 + \mathbf{s_2}) \right)
 dV(\mathbf{s_1})dV(\mathbf{s_2})
  \end{split}
  \end{equation*}

Since, $O_n^C \subset P_n \times P_n$, we can bound the second
integral again as shown below.
\begin{equation*}
 \begin{split}
V_{\oplus} \leq & \int_{P_n^*} M_{\oplus}(\Lambda_n,
\mathbf{s_1},\mathbf{s_2})
 dV(\mathbf{s_1})dV(\mathbf{s_2})  + \frac{V_{\oplus}}{4 V_{\oplus}'}\int_{P_n\times P_n}
M_{\oplus}'(\Lambda_n, \mathbf{s_1},\mathbf{s_2})
 dV(\mathbf{s_1})dV(\mathbf{s_2})\\
 & + \int_{(F_n\times G_n)^C}
\left( \sum_{h_1 \in \Lambda_n} \sum_{h_2 \in \Lambda_n}
\chi_{T}(h_1 + \mathbf{s_1}) \chi_T (h_2 + \mathbf{s_2}) \right)
 dV(\mathbf{s_1})dV(\mathbf{s_2})
 \end{split}
  \end{equation*}

By using Blichfeldt's principle and following similar steps as in
Corollary 4, we can show

 $ \int_{P_n\times P_n}
M_{\oplus}'(\Lambda_n, \mathbf{s_1},\mathbf{s_2})
 dV(\mathbf{s_1})dV(\mathbf{s_2}) = V_{\oplus}'$.

 \begin{equation*}
  \begin{split}
V_{\oplus} \leq & \int_{P_n^*} M_{\oplus}^{\Delta}(\Lambda_n,
\mathbf{s_1},\mathbf{s_2})
 dV(\mathbf{s_1})dV(\mathbf{s_2})  + \frac{V_{\oplus}}{4 } \\
 & + \int_{(F_n\times G_n)^C}
\left( \sum_{h_1 \in \Lambda_n} \sum_{h_2 \in \Lambda_n}
\chi_{T}(h_1 + \mathbf{s_1}) \chi_T (h_2 + \mathbf{s_2}) \right)
 dV(\mathbf{s_1})dV(\mathbf{s_2})
 \end{split}
 \end{equation*}

Next we change the double summation into a product of two
summations, to get
 \begin{equation*}
  \begin{split}
V_{\oplus} \leq & \int_{P_n^*} M_{\oplus}^{\Delta}(\Lambda_n,
\mathbf{s_1},\mathbf{s_2})
 dV(\mathbf{s_1})dV(\mathbf{s_2})  + \frac{V_{\oplus}}{4 }\\
 & + \int_{(F_n\times G_n)^C}
\left( \sum_{h_1 \in \Lambda_n} \chi_{T}(h_1 +
\mathbf{s_1})\right)\left( \sum_{h_2 \in \Lambda_n}
 \chi_T (h_2 + \mathbf{s_2}) \right)
 dV(\mathbf{s_1})dV(\mathbf{s_2})
 \end{split}
 \end{equation*}

The summation $  \sum_{h_1 \in \Lambda_n} \chi_{T}(h_1 +
\mathbf{s_1}) $ can be seen to count the number of lattice points in
the sphere $T_n$, for the translation $\mathbf{s_1}$. This is by
definition equivalent to $M(\Lambda_n , \mathbf{s_1})$. Similarly we
can replace the other summation by $M(\Lambda_n , \mathbf{s_2})$, to
get \begin{equation*}
  \begin{split}
V_{\oplus} \leq \int_{P_n^*} M_{\oplus}^{\Delta}(\Lambda_n,
\mathbf{s_1},\mathbf{s_2})
 dV(\mathbf{s_1})dV(\mathbf{s_2})  + \frac{V_{\oplus}}{4 }
  + \int_{(F_n\times G_n)^C}
M(\Lambda_n , \mathbf{s_1})M(\Lambda_n , \mathbf{s_2})
 dV(\mathbf{s_1})dV(\mathbf{s_2})
 \end{split}
 \end{equation*}
Since $(F_n \times G_n)^C \subset (F_n \times P_n)\cup(P_n\times
G_n)$ we can bound the integral to get
 \begin{equation*}
  \begin{split}
V_{\oplus} \leq & \int_{P_n^*} M_{\oplus}^{\Delta}(\Lambda_n,
\mathbf{s_1},\mathbf{s_2})
 dV(\mathbf{s_1})dV(\mathbf{s_2})
 + \frac{V_{\oplus}}{4 }
 + \int_{(F_n^C\times P_n)} M(\Lambda_n , \mathbf{s_1})M(\Lambda_n
, \mathbf{s_2})
 dV(\mathbf{s_1})dV(\mathbf{s_2}) \\
 &+ \int_{(P_n\times G_n^C)}
M(\Lambda_n , \mathbf{s_1})M(\Lambda_n , \mathbf{s_2})
 dV(\mathbf{s_1})dV(\mathbf{s_2})
 \end{split}
 \end{equation*}
 We next use Blichfeldt's principle to simplify the integrals.
\begin{equation*}
 \begin{split}
V_{\oplus} \leq &
 \int_{P_n^*}
M_{\oplus}^{\Delta}(\Lambda_n, \mathbf{s_1},\mathbf{s_2})
 dV(\mathbf{s_1})dV(\mathbf{s_2})
 + \frac{V_{\oplus}}{4 }
 + V_n(\sqrt{nP}) \int_{F_n^C} M(\Lambda_n , \mathbf{s_1})
 dV(\mathbf{s_1}) \\&
+ V_n(\sqrt{nP}) \int_{ G_n^C} M(\Lambda_n , \mathbf{s_2})
 dV(\mathbf{s_2})
 \end{split}
 \end{equation*}
 We use the condition $\mathbf{s_1} \not\in F_n$ and $\mathbf{s_2} \not\in G_n
 $, to get
\begin{equation*}
 \begin{split}
 V_{\oplus} \leq &
 \int_{P_n^*}
M_{\oplus}(\Lambda_n, \mathbf{s_1},\mathbf{s_2})
 dV(\mathbf{s_1})dV(\mathbf{s_2})
 + \frac{V_{\oplus}}{4 }  + V_n(\sqrt{nP}) \frac{V_n(\sqrt{nP})}{8 \det(\Lambda_n)} \int_{F_n^C}
 dV(\mathbf{s_1}) \\&
+ V_n(\sqrt{nP}) \frac{V_n(\sqrt{nP})}{8 \det(\Lambda_n)} \int_{
G_n^C}
 dV(\mathbf{s_2})
 \end{split}
\end{equation*}
Since $F_n^C, G_n^C \subset P_n$ and $\int_{P_n}
dV(\mathbf{s_1}),\int_{P_n} dV(\mathbf{s_2}) = \det (\Lambda_n)$, we
obtain,
\begin{equation*}
 \begin{split}
 V_{\oplus} \leq &
 \int_{P_n^*}
M_{\oplus}(\Lambda_n, \mathbf{s_1},\mathbf{s_2})
 dV(\mathbf{s_1})dV(\mathbf{s_2})
 + \frac{V_{\oplus}}{4 }  + V_n(\sqrt{nP}) \frac{V_n(\sqrt{nP})}{8 }
+ V_n(\sqrt{nP}) \frac{V_n(\sqrt{nP})}{8 }
 \end{split}
\end{equation*}
 Finally using $V_{\oplus} = (V_n(\sqrt{nP}))^2 $, we  obtain
\begin{equation*}V_{\oplus} \leq 2 \int_{P_n^*} M_{\oplus}(\Lambda_n,
\mathbf{s_1},\mathbf{s_2})
 dV(\mathbf{s_1})dV(\mathbf{s_2})\end{equation*}
\end{proof}

From the above lemma it can be seen that the measure of $P_n^*$ must
be non-zero and hence, there must be at least some translations
$(\mathbf{s_1},\mathbf{s_2})$ of the lattice, where the requirements
hold.

\section{Minkowski-Hlawka Theorem to Show Good Lattices Exist}

\begin{lem}\label{lem:existencelatticeMinkowski}
There exists translational vectors $\mathbf{s_1^*}$,
$\mathbf{s_2^*}$ such that the following holds
\begin{equation*}
\begin{split}
 \frac{1}{M_\oplus(\Lambda_n,\mathbf{s_1^*},\mathbf{s_2^*})}
\sum_{(\mathbf{x_1},\mathbf{x_2}) \in \mathcal{C}_{\oplus}^{\Delta}}
 \sum_{g  \in \Lambda_n \setminus
\{0 \}} p_\theta\left( \mathbf{x_1} + \mathbf{x_2},g + \mathbf{x_1}
+ \mathbf{x_2} \right)\chi_{T_{\sqrt{2P}}^{\Delta}}(g + \mathbf{x_1}
+ \mathbf{x_2})  \\ \leq 2
 \frac{(n-1)\pi^{\frac{n - 1}{2}}{(n(2P+\delta))^{n/2}}}{d_n n \Gamma (\frac{n +
 1}{2})}   \int_{0}^{\theta} \sin^{n-2}(x) dx
 \end{split}
\end{equation*}

and
\begin{equation*}
\begin{split} M_1(\Lambda_n,\mathbf{s_1^*}) \geq
\frac{V_n(\sqrt{nP})}{8 \det(\Lambda_n)};
M_2(\Lambda_n,\mathbf{s_2^*}) \geq \frac{V_n(\sqrt{nP})}{8
\det(\Lambda_n)};
\frac{M_\oplus'(\Lambda_n,\mathbf{s_1^*},\mathbf{s_2^*})}{M_\oplus(\Lambda_n,\mathbf{s_1^*},\mathbf{s_2^*})}\leq
4 \frac{V_\oplus'}{V_\oplus}
\end{split}
\end{equation*}

\end{lem}

\begin{proof}

 First let us define the function $f(z)$ as  follows
\begin{equation*}
\begin{split}
 f(\zv) = &  \int_{\mathds{R}^n} \int_{\mathds{R}^n}
  p_\theta\left(    \mathbf{u}  +  \mathbf{v},\zv + \mathbf{u} +
\mathbf{v} \right)  \chi_{T_{\sqrt{2P}}^{\Delta}}(g + \mathbf{u} +
\mathbf{v}) \chi_{T}( \mathbf{u} ) \chi_{T}( \mathbf{v} )
\chi_{T_{\sqrt{2P}}^{\Delta}}(  \mathbf{u} + \mathbf{v})
dV(\mathbf{u,v})
\end{split}
\end{equation*}
Clearly $f(z)$ is a nonnegative integrable function and has bounded
support. Hence we can apply the Minkowski-Hlawka theorem and find a
lattice $\Lambda_n^*$, such that
\begin{equation*}
\sum_{g  \in \Lambda_n \setminus \{0 \}} f(g) \leq \frac{1}{d_n}
\int_{\mathds{R}^n} f(\zv) dV(\zv)
\end{equation*}
Consider the integral
\begin{equation*}\begin{split}
I =& \int_{P_n^*} M_\oplus(\Lambda_n,\mathbf{s_1},\mathbf{s_2})
\left[\frac{1}{M_\oplus(\Lambda_n,\mathbf{s_1},\mathbf{s_2})}
\sum_{(\mathbf{x_1},\mathbf{x_2}) \in \mathcal{C}_{\oplus}^{\Delta}}
 \sum_{g  \in \Lambda_n \setminus
\{0 \}} p_\theta\left( \mathbf{x_1} + \mathbf{x_2}, g + \mathbf{x_1}
+ \mathbf{x_2} \right)\right.\\&\left.
\chi_{T_{\sqrt{2P}}^{\Delta}}(g + \mathbf{x_1} +
\mathbf{x_2})\right] dV(\mathbf{s_1,s_2})\end{split}
\end{equation*}
\begin{equation*}
I = \int_{P_n^*} \sum_{(\mathbf{x_1},\mathbf{x_2}) \in
\mathcal{C}_{\oplus}^{\Delta}}
 \sum_{g  \in \Lambda_n \setminus
\{0 \}} p_\theta\left( \mathbf{x_1} + \mathbf{x_2}, g +\mathbf{x_1}
+ \mathbf{x_2} \right)\chi_{T_{\sqrt{2P}}^{\Delta}}(g + \mathbf{x_1}
+ \mathbf{x_2}) dV(\mathbf{s_1,s_2})
\end{equation*}
\begin{equation*}
\leq \int_{P_n \times P_n} \sum_{(\mathbf{x_1},\mathbf{x_2}) \in
\mathcal{C}_{\oplus}^{\Delta}}
 \sum_{g  \in \Lambda_n \setminus
\{0 \}} p_\theta\left( \mathbf{x_1} + \mathbf{x_2}, g + \mathbf{x_1}
+ \mathbf{x_2} \right)\chi_{T_{\sqrt{2P}}^{\Delta}}(g + \mathbf{x_1}
+ \mathbf{x_2}) dV(\mathbf{s_1,s_2})
\end{equation*}
The above follows as $P_n^* \subseteq P_n \times P_n$.
\begin{equation*}
\begin{split}
& I \leq \int_{P_n \times P_n} \sum_{(\mathbf{h_1},\mathbf{h_2}) \in
\Lambda_n \times \Lambda_n}
 \sum_{g  \in \Lambda_n \setminus
\{0 \}} p_\theta\left( h_1 + \mathbf{s_1} + h_2 +  \mathbf{s_2},g +
h_1 + \mathbf{s_1} + h_2 + \mathbf{s_2} \right) \\&
\chi_{T_{\sqrt{2P}}^{\Delta}}(g + h_1 + \mathbf{s_1} + h_2 +
\mathbf{s_2}) \chi_{T}(h_1 + \mathbf{s_1} ) \chi_{T}(h_2 +
\mathbf{s_2} )\chi_{T_{\sqrt{2P}}^{\Delta}}( h_1 + \mathbf{s_1} +
h_2 + \mathbf{s_2}) dV(\mathbf{s_1,s_2})
\end{split}
\end{equation*}

In the above equation we substituted $\mathbf{x_1} = h_1 +
\mathbf{s_1}$ and $\mathbf{x_2} = h_2 + \mathbf{s_2}$.  Next
applying the Blichfeldt's principle twice we obtain,
\begin{equation*}
\begin{split}
 I \leq &\int_{\mathds{R}^n} \int_{\mathds{R}^n}
 \sum_{g  \in \Lambda_n \setminus
\{0 \}} p_\theta\left( \mathbf{u}  +  \mathbf{v}, g  + \mathbf{u} +
\mathbf{v} \right)  \chi_{T_{\sqrt{2P}}^{\Delta}}(g + \mathbf{u} +
\mathbf{v}) \chi_{T}( \mathbf{u} ) \chi_{T}( \mathbf{v} )
\chi_{T_{\sqrt{2P}}^{\Delta}}(  \mathbf{u} + \mathbf{v})
dV(\mathbf{u,v})
\end{split}
\end{equation*}

Here $\uv = h_1 + \mathbf{s_1}$ and $\vv = h_2 + \mathbf{s_2}$. We
can next take the summation outside the double integral, as we are
dealing with a finite number of non-zero sums. This gives,
\begin{equation*}
\begin{split}
 I \leq & \sum_{g  \in \Lambda_n \setminus
\{0 \}} \int_{\mathds{R}^n} \int_{\mathds{R}^n}
  p_\theta\left(  \mathbf{u}  +  \mathbf{v},g  + \mathbf{u} +
\mathbf{v} \right)  \chi_{T_{\sqrt{2P}}^{\Delta}}(g + \mathbf{u} +
\mathbf{v}) \chi_{T}( \mathbf{u} ) \chi_{T}( \mathbf{v} )
\chi_{T_{\sqrt{2P}}^{\Delta}}(  \mathbf{u} + \mathbf{v})
dV(\mathbf{u,v})
\end{split}
\end{equation*}

Next we use the Minkowski-Hlawka theorem to establish that there
exists at least one lattice $\Lambda_n^*$ such that the summation
can be bounded by an integral, as shown below. In short the
Minkowski-Hlawka theorem gives an existence result that such a
lattice can be found, but does not give a way to find such a
lattice.
\begin{equation*}
\begin{split}
 I \leq  \frac{1}{d_n} \int \int_{\mathds{R}^n} \int_{\mathds{R}^n}
  p_\theta\left(  \mathbf{u} +
\mathbf{v}, \mathbf{w}   \right)
\chi_{T_{\sqrt{2P}}^{\Delta}}(\mathbf{w} ) \chi_{T}( \mathbf{u} )
\chi_{T}( \mathbf{v} ) \chi_{T_{\sqrt{2P}}^{\Delta}}( \mathbf{u} +
\mathbf{v}) dV(\mathbf{u,v})dV(\mathbf{w})
\end{split}
\end{equation*}

In the above we have replaced $\mathbf{w} + \uv + \vv$ with
$\mathbf{w}$. Again we change the order of integration to get,

\begin{equation*}
\begin{split}
 I \leq  \frac{1}{d_n} \int_{\mathds{R}^n} \int_{\mathds{R}^n} \left[ \int
  p_\theta\left(  \mathbf{u} +
\mathbf{v}, \mathbf{w}    \right)
\chi_{T_{\sqrt{2P}}^{\Delta}}(\mathbf{w} )dV(\mathbf{w})\right]
\chi_{T}( \mathbf{u} ) \chi_{T}( \mathbf{v} )
\chi_{T_{\sqrt{2P}}^{\Delta}}( \mathbf{u} + \mathbf{v})
dV(\mathbf{u,v})
\end{split}
\end{equation*}

The inner integral can be seen to be independent of $\uv + \vv$, and
the outer integral can be seen to be lesser than $V_{\oplus}$. Hence
we can express the resultant integral as,
\begin{equation*}
\begin{split}
 I \leq  \frac{V_{\oplus}}{d_n}  \int
  p_\theta\left(  \mathbf{s_0} ,  \mathbf{w}   \right)  \chi_{T_{\sqrt{2P}}^{\Delta}}(\mathbf{w}
)dV(\mathbf{w}).
\end{split}
\end{equation*}

Hence we next evaluate the integral using similar steps as in
\cite{urbanke98}, first by changing to spherical coordinates.
\begin{equation*}
\begin{split}
 I \leq  \frac{V_{\oplus}}{d_n}
 \int_{\sqrt{n(2P-\delta)}}^{\sqrt{n(2P+\delta)}} \left[\int_{\mathbf{w}:\|\mathbf{w}\| = r}
  p_\theta\left(  \mathbf{s_0} ,  \mathbf{w}   \right)  \chi_{T_{\sqrt{2P}}^{\Delta}}(\mathbf{w}
)dA(\mathbf{w}) \right] dr,
\end{split}
\end{equation*}
where $dA(\mathbf{w})$ is the $(n-1)$-dimensional volume element of
a thin spherical element at a distance $r$ from the origin.  Next we
use the definition of $p_\theta$ to get,

\begin{equation*}
\begin{split}
 I \leq  \frac{V_{\oplus}}{d_n}
 \int_{\sqrt{n(2P-\delta)}}^{\sqrt{n(2P+\delta)}} \left[\int_{\mathbf{w}:\|\mathbf{w}\| = r}
    \mbox{Pr}\left( \pi( \mathbf{s_0}) +
Z^n \in B_{\theta}(\mathbf{w}) \right)
\chi_{T_{\sqrt{2P}}^{\Delta}}(\mathbf{w} )dA(\mathbf{w}) \right] dr.
\end{split}
\end{equation*}

We evaluate the probability by conditioning on $Z^n$.
\begin{equation*}
\begin{split}
 I \leq  \frac{V_{\oplus}}{d_n}
 \int_{\sqrt{n(2P-\delta)}}^{\sqrt{n(2P+\delta)}} \left[\int_{\mathbf{w}:\|\mathbf{w}\| =
 r} \left[
    \int \mbox{Pr}\left( \pi( \mathbf{s_0}) +
Z^n \in B_{\theta}(\mathbf{w}) | Z^n = \zv \right) \mbox{Pr}(Z^n =
\zv) dV(\zv) \right] \right. \\ \left.
\chi_{T_{\sqrt{2P}}^{\Delta}}(\mathbf{w} )dA(\mathbf{w}) \right] dr.
\end{split}
\end{equation*}

Since the function is non-negative, the order of integration can be
interchanged to obtain

\begin{equation*}
\begin{split}
 I \leq  \frac{V_{\oplus}}{d_n} \int
 \left[ \int_{\sqrt{n(2P-\delta)}}^{\sqrt{n(2P+\delta)}} \int_{\mathbf{w}:\|\mathbf{w}\| =
 r} \left[
     \mbox{Pr}\left( \pi( \mathbf{s_0}) +
Z^n \in B_{\theta}(\mathbf{w}) | Z^n = \zv \right) \right] \right. \\
\left. \chi_{T_{\sqrt{2P}}^{\Delta}}(\mathbf{w} )dA(\mathbf{w}) dr
\right] \mbox{Pr}(Z^n = \zv) dV(\zv).
\end{split}
\end{equation*}

Due to the conditioning on $Z^n$, the conditional probability
becomes deterministic and is equivalent to the cross sectional area
of a hyper cone of half-angle $\theta$, intersecting with a
hyper-sphere of radius $r$. Note that in  the result in
\cite{urbanke98}, the area of cross section must be a function of
$r$ and not of $\sqrt{nS'}$. This gives,

\begin{equation*}
\begin{split}
 I \leq  \frac{V_{\oplus}}{d_n} \left[\int \mbox{Pr}(Z^n = \zv) dV(\zv)\right]
 \int_{\sqrt{n(2P-\delta)}}^{\sqrt{n(2P+\delta)}} \left[\frac{(n-1)\pi^{\frac{n - 1}{2}}{r^{n-1}}}{\Gamma (\frac{n +
 1}{2})}\int_{0}^{\theta} \sin^{n-2}(x) dx
  \right] dr  .
\end{split}
\end{equation*}

This in turn yields
\begin{equation*}
\begin{split}
 I \leq  \frac{V_{\oplus}}{d_n}
 \left[\frac{(n-1)\pi^{\frac{n - 1}{2}}{\left((n(2P+\delta))^{n/2}- (n(2P-\delta))^{n/2}\right)}}{n \Gamma (\frac{n +
 1}{2})}   \int_{0}^{\theta} \sin^{n-2}(x) dx
  \right].
\end{split}
\end{equation*}

\begin{equation*}
\begin{split}
 I \leq  \frac{V_{\oplus}}{d_n}
 \left[\frac{(n-1)\pi^{\frac{n - 1}{2}}{(n(2P+\delta))^{n/2}}}{n \Gamma (\frac{n +
 1}{2})}   \int_{0}^{\theta} \sin^{n-2}(x) dx
  \right].
\end{split}
\end{equation*}

Next using Lemma~\ref{lem:blichtranslationrate} we can bound the
resultant integral to get,

\begin{equation*}
\begin{split}
 I \leq & 2 \int_{P_n^*} M_{\oplus}(\Lambda_n , \mathbf{s_1},\mathbf{s_2}) \left[
 \frac{(n-1)\pi^{\frac{n - 1}{2}}{(n(2P+\delta))^{n/2}}}{d_n n \Gamma (\frac{n +
 1}{2})}   \int_{0}^{\theta} \sin^{n-2}(x) dx\right] dV(\mathbf{s_1},\mathbf{s_2})
  .
\end{split}
\end{equation*}
Now comparing the integrals and also knowing that the measure of
$P_n^*$ is non-zero we can see that there must be at least one such
translational vector pair $(\mathbf{s_1},\mathbf{s_2})$ such that
the required assertion of the lemma holds.
\end{proof}

\end{document}